\documentclass[
reprint,
groupedaddress,
nofootinbib,
 amsmath,amssymb,
 aps, 
]{revtex4-2}

\usepackage{physics}

\usepackage{graphicx}
\usepackage{dcolumn}
\usepackage{bm}
\usepackage[colorlinks=true,urlcolor=blue,linkcolor=blue,citecolor=magenta]{hyperref}
\usepackage[normalem]{ulem}
\usepackage{tikz}
\usetikzlibrary{calc, backgrounds}

\usetikzlibrary{
  graphs,
  graphs.standard
}
\usepackage{MnSymbol,wasysym}

\usepackage{amsmath, amssymb, amsthm, mathtools}
\usepackage[T1]{fontenc}
\usepackage{bbold}
\usepackage{dsfont}
\usepackage{standalone}

\tikzset{
    treenode/.style={circle, draw=black, minimum size=0.5cm, inner sep=0pt},
    treeedge/.style={thick}
}

\newcommand{\edgepair}[2]{%
    \draw[treeedge]
        ($ (#1)!2pt!90:(#2) $) -- ($ (#2)!2pt!-90:(#1) $);
    \draw[treeedge, dashed]
        ($ (#1)!2pt!-90:(#2) $) -- ($ (#2)!2pt!90:(#1) $);
}

\newcommand{\edgequad}[2]{%
    \foreach \off/\style in {-4.5/solid, -1.5/dashed, 1.5/solid, 4.5/dashed} {
        \draw[treeedge, \style]
            ($ (#1)!\off pt!90:(#2) $) --
            ($ (#2)!\off pt!-90:(#1) $);
    }
}

\newcommand{\edgedouble}[2]{%
    \draw[treeedge]
        ($ (#1)!2pt!90:(#2) $) -- ($ (#2)!2pt!-90:(#1) $);
    \draw[treeedge]
        ($ (#1)!2pt!-90:(#2) $) -- ($ (#2)!2pt!90:(#1) $);
}

\newcommand{\edgegray}[2]{%
    \draw[treeedge, gray!50, opacity=0.4] (#1) -- (#2);
}

\usepackage{xcolor}
\usetikzlibrary{calc}

\newcommand{\mysection}[1]{\paragraph*{#1 ---}}

\usepackage{comment}
\usepackage{amsthm}

\usepackage{soul}
\usepackage[normalem]{ulem}
\usepackage{comment}
\usepackage{breakcites}
\usepackage[capitalize,nameinlink]{cleveref}

\usepackage[sanserif,basic]{complexity}
\usepackage{scalerel}
\usepackage[colorinlistoftodos]{todonotes}
\usepackage{algorithm}
\usepackage{setspace}
\usepackage{algpseudocode}
\usepackage{xspace}
\usepackage{extarrows}
\usepackage{environ}
\usepackage{needspace}
\usepackage{framed}
\usepackage{enumitem} 
\theoremstyle{definition}

\newtheorem{prototheorem}{Theorem}
\newtheorem{theorem}[prototheorem]{Theorem}
\newtheorem{lemma}[prototheorem]{Lemma}
\newtheorem{definition}[prototheorem]{Definition}
\newtheorem{conjecture}[prototheorem]{Conjecture}

\colorlet{theoremshade}{cyan!15}
\colorlet{lemmashade}{green!15} %

\newenvironment{theo}{\colorlet{shadecolor}{theoremshade}\begin{shaded}\begin{theorem}}
{\end{theorem}\end{shaded}}

\newenvironment{lem}{\colorlet{shadecolor}{lemmashade}\begin{shaded}\begin{lemma}}
{\end{lemma}\end{shaded}}

\newenvironment{defn}{\colorlet{shadecolor}{gray!15}\begin{shaded}\begin{definition}}
{\end{definition}\end{shaded}}

\usepackage{nicefrac}

\makeatletter
\def\l@subsubsection#1#2{}
\makeatother

\begin{document}

\title{Hierarchical Prototype Emergence in Modern Hopfield Models}
\author{Aditya Cowsik}
\thanks{Aditya Cowsik and Adithya Sriram contributed equally to this work.}
\author{Adithya Sriram}
\thanks{Aditya Cowsik and Adithya Sriram contributed equally to this work.}
\affiliation{
 Department of Physics, Stanford University, Stanford, CA 94305, USA
}
\date{\today}

\begin{abstract}
Hierarchical correlations are a universal feature of any realistic model of data, and the question of how associative memory models may learn these correlations and generalize beyond them to construct new sensible images is an important step towards understanding more complex modern architectures such as diffusion models. We consider a hierarchical model for memories which are sampled and stored in a dense Hopfield network with polynomial activation. We analytically derive conditions for each level of this hierarchy to be locally stable – that is they are local energy minima. We use prototype reconstruction as a minimal model of generalization and we find that it takes only a quasi-polynomial amount of information to generalize beyond particular memories and even particular groups
in the hierarchy. We observe a qualitatively analogous phase diagram in
the number of memories, sharpness of the activation function (polynomial degree)
for data from Fashion-MNIST.
\end{abstract}

\maketitle

\mysection{Introduction}

The advent of transformer based language models has revolutionized generative artificial intelligence technologies \cite{vaswani2023attentionneed, devlin-etal-2019-bert, dosovitskiy2021imageworth16x16words, peebles2023scalablediffusionmodelstransformers}. Despite the incredible progress in their capabilities, exactly \textit{how} these architectures learn and are capable of generating sensible content is an area of intense scrutiny \cite{simon2026scientifictheorydeeplearning}. In the specific context of diffusion models and image generation, how the model is able to generalize images from its training set and create entirely new and creative out-of-distribution images from a simple query is a fascinating question, and mechanistic descriptions of this process are just beginning to emerge \cite{sohldickstein2015deepunsupervisedlearningusing, kamb2024analytic, hunt2026exactinformationtheorygeneralization, pham2026memorizationgeneralizationemergencediffusion, niedoba2025mechanisticexplanationdiffusionmodel, cui2025solvablemodellearninggenerative}.

At the same time as these novel architectures emerge, Hopfield networks and associative memory models have also seen a number of interesting innovations \cite{doi:10.1073/pnas.79.8.2554, krotov2025modernmethodsassociativememory, ramsauer2021hopfieldnetworksneed, delgaudio2026shorttermplasticityrecallsforgotten, lufkin2026hybridassociativememories, Steinberg2022-jd}. These models are content-addressable self correcting memories and the original Hopfield networks were shown to be able to reliably store a number of memories which scaled linearly in the number of neurons \cite{doi:10.1073/pnas.79.8.2554}. More recently, dense associative memories have emerged as a modern generalization of the original Hopfield model \cite{krotov2016dense}. In these networks, the quadratic energy function is replaced with a higher order polynomial or exponential activation function. This sharpens the basin of attraction around the stored memories and boosts the storage capacity to super-linear or even exponential in the number of neurons \cite{krotov2016dense, PhysRevLett.132.077301, kafraj2026biologicallyplausibledenseassociative}. 
\begin{figure}
    \centering
    \includegraphics[width=0.8\linewidth]{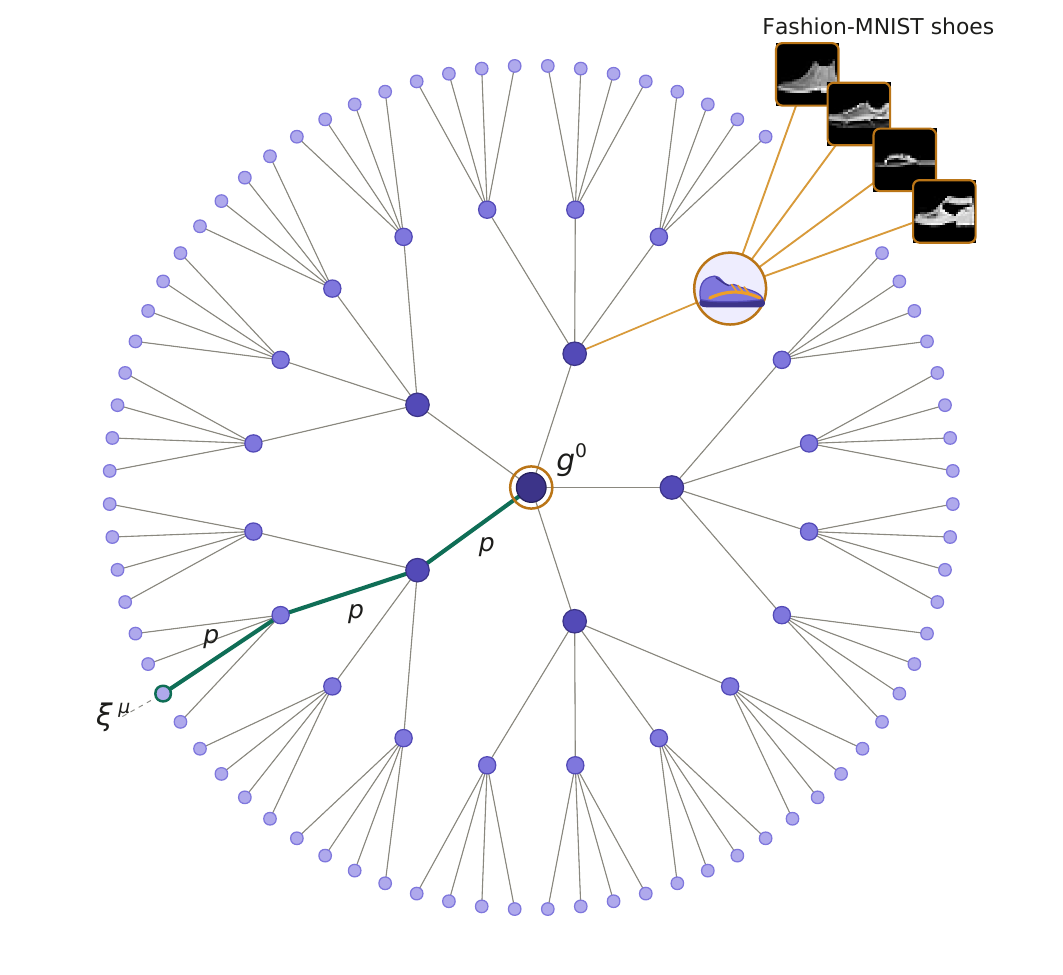}
    \caption{Schematic of the hierarchical memory structure we consider in this work. Patterns seen by the network $\boldsymbol{\xi}^\mu$ are the leaves (pictures of shoes) and generalized patterns $\mathbf{g}^\ell$ (reconstructed prototypes) are those which are further up in the correlation network. In the toy model we consider, patterns are derived from their direct ancestors by applying noise to a fraction $p$ of the bits of the ancestor pattern (\cref{eq:noise}).  %
    }
    \label{fig:cartoon1}
\end{figure}

While in diffusion models, the task is to generalize from the training set and produce new patterns, the task in associative memory models is exactly the opposite. Yet there is an intriguing analytical correspondence between the two, as it was shown that diffusion models and dense associative memories have related energy functions and similar properties as dynamical systems \cite{ambrogioni2023search, pham2026memorizationgeneralizationemergencediffusion, niu2024scalinglawsunderstandingtransformer}. This motivates a perspective of interpolating between these two and exploring a \textit{memorization-generalization transition}. Understanding the structure of Hopfield networks could help us understand when they merely reproduce memorized data, and when they can generalize beyond what they have already seen. This question is closely related to the notion of capacity in generalized Hopfield models. Indeed, the spurious minima generated in the energy landscape of dense associative memories are being explored as a mechanism behind generalization \cite{pham2026memorizationgeneralizationemergencediffusion, hoover2026denseassociativememoryepanechnikov}. However, it is insufficient to say that generalization happens precisely when we exceed a capacity threshold.

We adopt this perspective and explore a model of dense associative memory in which \textit{the encoded data contains a latent hierarchical correlation structure.} Data with latent hierarchical structure is very common. This type of data can be viewed as storing "example" patterns, such as images of shoes. These example patterns would have emerged from more generalized "concepts" which sit further up in the correlation structure. The images of shoes are examples of a generalized ``footwear concept" which is itself an example of a generalized ``clothing concept" (see ~\cref{fig:cartoon1}). %

Under the assumption of this latent correlational structure, and using prototype reconstruction as a minimal model for generalization, we attempt to answer the following two questions: 1) With hierarchically correlated memories, when do dense associative memory models memorize and remember patterns? 2) Can these dense associative memory models recover the generalized patterns from the underlying correlation structure? We find that not only can the network reliably recover the stored patterns, but new ``spurious'' patterns emerge as minima in the energy landscape and many of these patterns are exactly those which are further up in the correlation structure. This competition between memorization and generalization can be tuned by the network parameters. Our work sheds light on how hierarchical features in data can be learned and how that is precisely related to memorization/forgetting in an exactly solvable model.%

Before continuing, we review some related work. We first note that Hopfield networks with correlated patterns have been studied in previous work, though these models consider only quadratic activation functions, whereas the dense activation function is the crucial ingredient of our work~\cite{DOTSENKO1986410, engelJPA, Agliari:2013wo, sargolzaei2025hierarchical}. In the context of dense associative memories, the idea of hierarchical organization has begun to receive interest. Ref. \cite{krotov2021hierarchicalassociativememory} studied models where the architecture itself involved hierarchically organized hidden layers of dense associative memories. Other related work by \cite{Agliari2023-an, agliari2021emergenceconceptshallowneural} and concerns the use of a neural networks in retrieving an archetype memory or concept given many training examples. Our contribution beyond these is to consider arbitrary depth trees and analyze the relative stability of memories at different tree levels. We also show how new memories and patterns may be generated by the spurious minima. 

\mysection{Model Description}

We consider a system of binary neurons, each of which is denoted by a variable $\sigma_i$ which can take on values $\pm 1$ \citep{doi:10.1073/pnas.79.8.2554}. The state of the entire system is denoted $\boldsymbol{\sigma} \in \{-1, 1\}^N$. A pattern to be stored, or memory, is denoted as $\boldsymbol{\xi}$, where the $i\textsuperscript{th}$ index $\xi_i$ is the state of the $i\textsuperscript{th}$ neuron in the memory. We define the following as the energy function for the system:
\begin{align} \label{eq: activation function}
   E(\boldsymbol\sigma) = -\sum_{\mu=1}^M F(\boldsymbol{\xi}^\mu \cdot \boldsymbol{\sigma}),
\end{align}
where $F(x)$ is an activation function which here takes in as input the dot product between the memory and the current state of the system. Here, we will consider polynomial activation functions, i.e. $F(x) = x^n$. Also, $M$ is the total number of encoded memories. 
Recovery of memories happens by performing local descent starting at a probe point $\boldsymbol{\sigma}^0$ until a fixed point (local minimum) is reached. 

 In the toy model of data we consider, the memories are correlated in a tree structure, see ~\cref{fig:cartoon1}. The central root prototype is denoted in the figure as $\mathbf{g}^0$. Derived from this pattern are successive layers of prototypes, each generating the next layer. This continues arbitrarily many times until finally one reaches the bottom layer of the tree. In our model, this tree is taken to be of depth $h$. From each prototype, we derive $K-1$ memories and so the tree is degree $K$ (from the root we derive $K$ memories). We take each derived memory to be generated by slightly corrupting its prototype with a finite density of noise $\mathbf{e}$. This noise is modeled with a random vector where each entry $e_i \in \{\pm 1\}$ is i.i.d. and $\mathbb{P}(e_i = -1) = p$. We restrict to $0 < p < 1/2$. Then, for example, a level $\ell$ prototype $\mathbf{g}^\ell$, which is derived from a level $\ell - 1$ prototype $\mathbf{g}^{\ell - 1}$ takes the form
\begin{align} \label{eq:noise}
    \mathbf{g}^\ell = \mathbf{g}^{\ell - 1} \circ \mathbf{e}
\end{align}
where here $\circ$ denotes element-wise multiplication, i.e. $g^\ell_i = g^{\ell-1}_i e_i$. For simplicity we maintain a uniform correlation between prototypes and derived memories and a uniform degree, i.e. $p$ and $K$ are uniform throughout the tree. Importantly, we initialize the network \textit{with only leaf memories $\mathbf{\xi}$}. That is, in the energy function ~\cref{eq: activation function}, the network only penalizes states based on their Hamming distance from the leaf memories. 

We remark on one important point. Often times in Hopfield models, the patterns represent bit configurations, such as the states of pixels in a black and white image. In that case, if a bit in the state is flipped "on", then the corresponding pixel is lit up. The model we present could indeed correspond to such a system and is a natural way of viewing things, but it could also correspond to something more general. For instance, two images depicting similar objects may look completely uncorrelated from a pixel-to-pixel point of view, but are in fact highly correlated in some latent space of objects. Our model is intended to capture these types of correlational structures as well, if the right variables are used as inputs.

\mysection{Memory Stability and Emergent Minima} 
For any pattern, its stability to perturbations is given by the following \textit{energy gap}:
\begin{align} \label{eq:gap}
    \Delta E (\boldsymbol{\sigma}) = E(\boldsymbol{\sigma} - 2 \sigma_i \hat{\mathbf{e}}_i)-E(\boldsymbol{\sigma}).
\end{align}
This quantity measures the gap in energy between the state of the system $\boldsymbol{\sigma}$ and the state with the bit at position $i$ flipped. If the energy gap is positive for all possible flips, then the pattern is indeed a local minima. The statistics of this gap inform us as to which patterns are local energy minima under the random ensemble in consideration. In particular, we are interested in whether the mean energy gap is positive, i.e. $\langle \Delta E \rangle > 0$, and in the squared coefficient of variation, $(\mathcal{C}^v)^2$, i.e. ratio between the energy gap variance and the mean squared. This is akin to an inverse signal to noise ratio and it sets the probability that a pattern has a negative energy gap. This analysis yields the familiar linear memory capacities of traditional Hopfield models and the superlinear capacities of dense associative memories \cite{doi:10.1073/pnas.79.8.2554, Geszti1990-yk, krotov2016dense,PhysRevLett.132.077301}. 
\begin{figure}
    \centering
    \includegraphics[width=.99\linewidth]{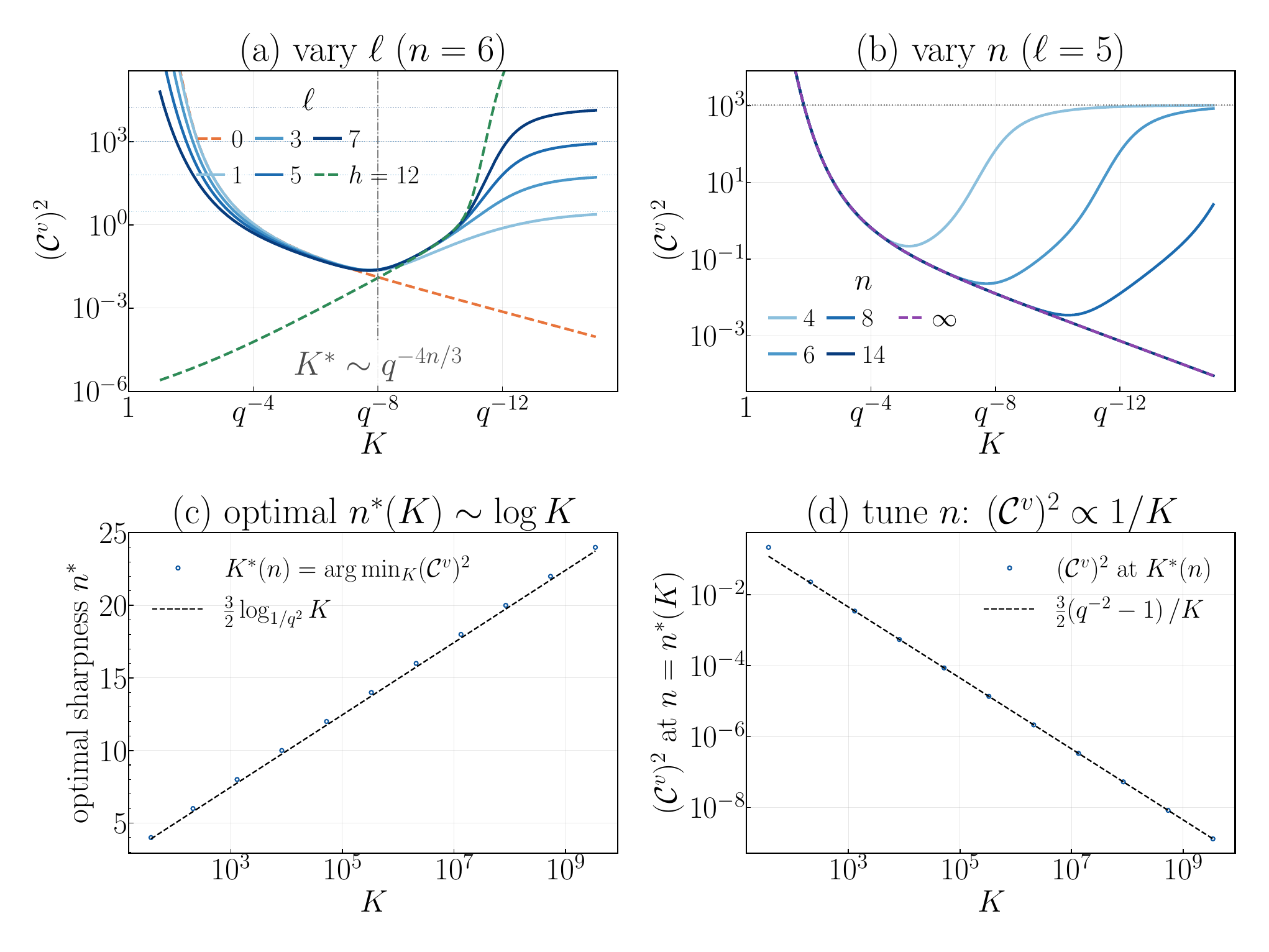}
    \caption{The squared coefficient of variation, $(\mathcal{C}^v_\ell)^2$, can be made asymptotically small at all $\ell$ simultaneously with the appropriate choice of sharpness, $n$. {\textbf{(a)}} Varying $\ell$ at a fixed $n=6$ demonstrates how the stability in the intermediate $K$ regime is largely $\ell$ independent, but beyond the crossover, patterns closer to the leaves become highly unstable. At the crossover all levels are approximately equally stable. {\textbf{(b)}} Varying the activation power $n$ at fixed depth ($\ell = 5$) causes the coefficient of variation to diverge when $K<q^{-2}$, reach a minimum for larger $K$ before saturating as $K \to \infty$. {\textbf{(c)}} Optimal value of $n$, $n^*$ vs $K$. {\textbf{(d)}} Tuning $n$ to the optimal value $n^*(K)$ (as shown in {\textbf{(c)}}), which minimizes $(\mathcal{C}^v_\ell)^2$ given other variables are fixed, leads to $(\mathcal{C}^v_\ell)^2 \sim 1/K$ at all depths (here shown for $\ell = 5$).}
    \label{fig:setting3}
\end{figure}

When $\boldsymbol{\sigma}$ corresponds to an encoded leaf memory $\xi$, stability is synonymous with \textit{memorization} -- the network can recover a stored pattern. On the other hand, the second question posed in the introduction pertains to whether the model can understand and recover the underlying correlational structure in the data. In this case, we can assess the energy gap statistics for when $\boldsymbol{\sigma}$ is an ancestor prototype, e.g. $\mathbf{g}^{\ell}$. Stability of $\mathbf{g}^\ell$ to perturbations is prototype reconstruction and is what we consider to be \textit{generalization}.

We now discuss evaluation of the statistics of ~\cref{eq:gap} for various $\boldsymbol{\sigma}$. There are three types of memories whose stabilities in which we are interested:  1) ancestor prototypes $\mathbf{g}^{\ell}$ where $0 < \ell < h$, 2) the leaf memories $\boldsymbol{\xi}^\mu$, and 3) the root $\mathbf{g}^0$. As mentioned earlier, the stabilities of these states are governed by the sign of the mean, as well as  $(\mathcal{C}^v_{\ell})^2$. Here, the subscript $\ell$ denotes the level of the prototype and due the symmetry of the model, all prototypes of the same level have the same statistics. 

For the parameter regime $q^{2h} = \omega(n/\sqrt{N})$, the mean and the variance have the following forms for any probe point which is a tree node:
\begin{align} \label{eq: generic moments}
    \langle \Delta E \rangle_\ell &= 2n N^{n-1} \sum_{d=h-\ell}^{h+\ell} C^{(1)}_\ell(d) q^{n d} \\
    \text{Var}[(\Delta E)_\ell] &= 4n^2 N^{2n-2} \sum_{(s,t_1,t_2)} C^{(2)}_{\ell}((s,t_1,t_2)) \nonumber \\
    &\times q^{n(2s+t_1+t_2) -2s}(1-q^{2s}).
\end{align}
In these expressions, $q = 1-2p$. In the second line, $s,t_1,t_2$ are the distances from the probe $\boldsymbol{\sigma}$, target 1 and target 2 respectively to their median point in the tree. The $C_\ell^{(\cdot)}$ are combinatorial coefficients. Within the parameter regime described above, these expressions are valid up to an $o(1)$ relative error. Further evaluation of the above equations and determination of the $C_\ell^{(\cdot)}$ requires sifting through many tedious combinatorics. To organize these computations, we developed a diagrammatic formalism in analogy with how Feynman diagrams help organize the combinatorial problem of perturbation theory in field theory. We describe this formalism in detail, derive the above equations, and establish rigorous control over the approximation errors in the supplementary material. 

From ~\cref{eq: generic moments}, as the combinatorial coefficients $C^{(1)}_\ell$ are always positive, the mean energy gap is strictly positive for any set of parameter values. Therefore the stabilities are entirely determined by the coefficients of variation. 

 We first discuss the case of the ancestor memories, where $\ell$ is in the range $0 < \ell < h$. In ~\cref{fig:setting3}(a), we plot the analytical expression for $(\mathcal{C}^v_{0 < \ell < h})^2$ calculated for $0 < \ell < h$, in the large $h$ limit. Recall that this $\ell$ corresponds to ancestor memories which \textit{were never explicitly seen by the network}. 

We find three distinct regimes in the parameter space as a function of $K$, and its ratio to powers of $q$. First, when $Kq^2 < 1$, the ancestor memories at all $\ell$ are highly unstable, as evidenced by a $(\mathcal{C}^v_{0 < \ell < h})^2$ which quickly diverges as $K$ approaches 0. This is readily understood because the network does not have enough training examples to reconstruct the ancestor memories. Once $Kq^2 > 1$, the squared coefficient of variation becomes finite and falls to a minimum, before beginning to increase again as $K$ is further increased past $K > q^{-(2n-1)}$. Eventually, in the large $K$ limit, the squared coefficient of variation saturates to an $\ell$ dependent plateau. Notably, adding more data does \emph{not} improve recall unless the model capacity ($n$) is increased commensurately.  The height of this plateau diverges rapidly with $\ell$, indicating that the large number of training examples sufficiently interfere with one another and serve to "wash out" the differences between the different ancestor memories.  

The intermediate regime of $K$ between $q^{-2}$ and $q^{-(2n-1)}$ is the most interesting. Here, we observe that $(\mathcal{C}_{0<\ell<h}^v)^2$ decreases roughly as $\sim K^{-1}$ until a minimum point. For large $\ell$, the value of $K$ at which this minimum point occurs saturates to a value $K^* \propto q^{-4n/3}$ and inverting this relation yields the optimal $n^*$ ~\cref{fig:setting3}(c). Furthermore, the actual value of $(\mathcal{C}^v_\ell)^2$ at $K^*$ may be driven downwards with increasing $n$. In ~\cref{fig:setting3}(d) we increase $K$ and tune $n = n^*(K)$ and see that minimum value $(\mathcal{C}^v_\ell)^2$ decreases as roughly $\propto K^{-1}$. What we therefore find is that in this parameter regime, the probability for an ancestor prototype to be unstable is decreasing in $K$, and so long as $n$ is increased commensurately, this probability may be decreased indefinitely towards zero (minima in ~\cref{fig:setting3}(b)). 

This ability to suppress $(\mathcal{C}_{0<\ell<h}^v)^2$ with $K$ and $n$ indicates how we should scale these parameters with $N$. Indeed, if we take $K$ to be an appropriate superlinear function of $N$ and $n \propto \log N$, then $(\mathcal{C}_{0 < \ell < h}^v)^2$ will vanish with $N$. Subsequently, a straightforward application of Chebyshev's inequality reveals that the probability for an ancestor prototype to be unstable vanishes. Furthermore, the stated parameter regime ($q^{2h} = \omega(n/\sqrt{N})$) means that this ability to suppress $(\mathcal{C}_{0 < \ell < h}^v)^2$ is valid up to $h\propto \log N$. In other words, with $ N^{\Theta(\log N) + o(1)}$ amount of training examples, the dense Hopfield network is able to reconstruct prototypes and generalize up to a depth $ \log N$.

Next, let us discuss the stability of the leaf memories $\boldsymbol{\xi}^\mu$ and the root prototype $\mathbf{g}^0$, for which $\ell = h$ and $\ell = 0$ respectively. The leaf stability corresponds entirely to a recall phenomenon rather than that of generalization. In contrast to the ancestor prototypes, we find that $(\mathcal{C}^v_{\ell = h})^2$ monotonically increases with $K$ (green dashed line in ~\cref{fig:setting3}(a)), indicating that leaf memory stability is highest when $K$ is small. Unlike with the ancestor prototypes, the correlations between training examples damages the network's ability to recall them individually. On the other hand for the root prototype, $(\mathcal{C}^v_{\ell = 0})^2$ decreases monotonically in $K$ (orange dashed line in ~\cref{fig:setting3}(a)). In this case, every training example reinforces the root, so its stability is only improved by an increase in $K$.

\begin{figure}
    \centering
    \includegraphics[width=0.95\linewidth]{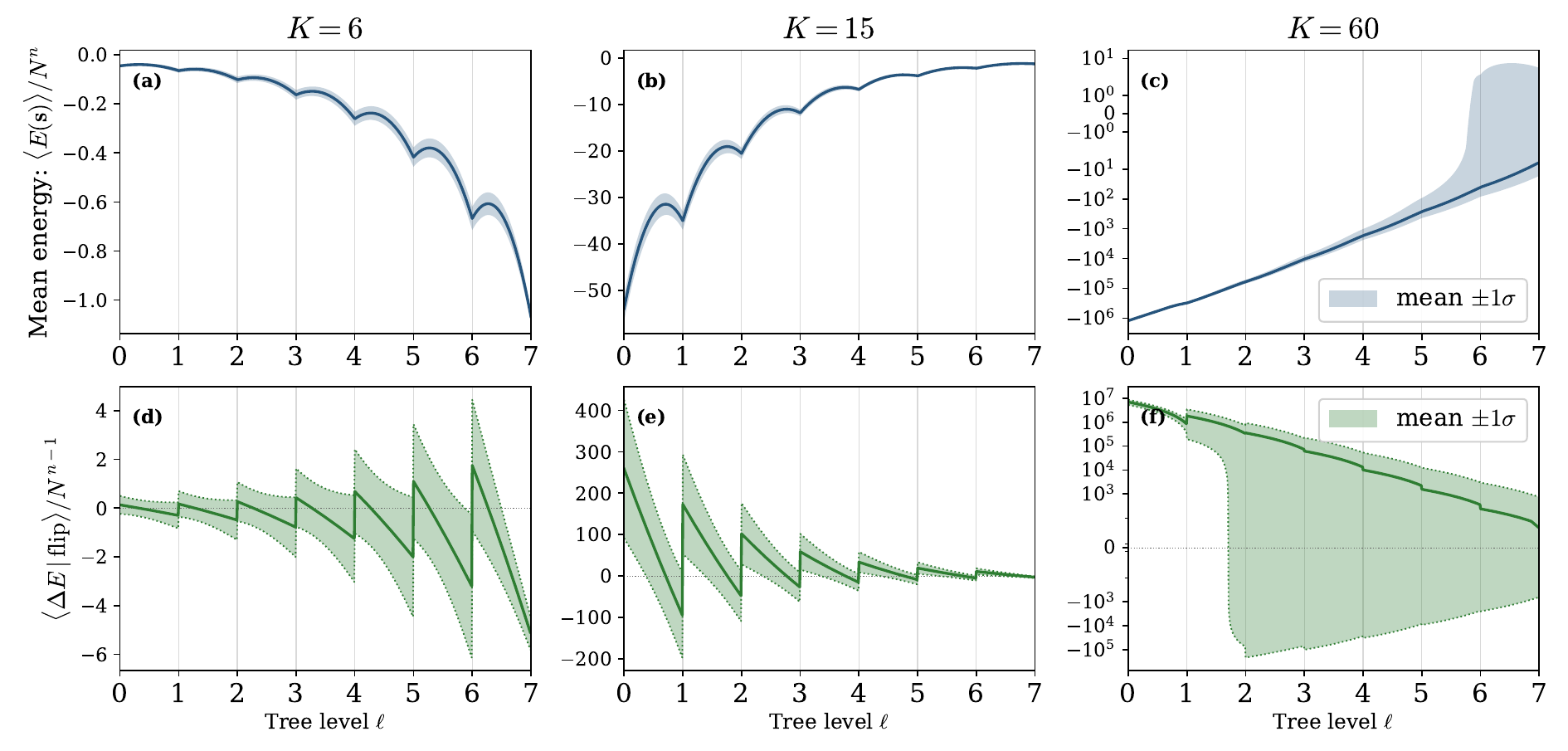}
    \caption{Energy landscape of a hierarchical Hopfield network along a root-to-leaf   walk, computed to leading order in $N$ from the exact tree moments.  The probe traverses each parent--child edge along a shortest Hamming path. Integer positions $\ell$ correspond to the prototypes patterns. The top row shows the mean energy while the bottom row shows the mean energy cost of the next spin flip of the walk.  Shaded regions (dotted edges) denote $\pm1\sigma$ fluctuations over the disorder ensemble. The $\langle E\rangle$ band scales as $N^{-1/2}$ (shown for $N=10^{3}$), while the gap band is $N$-independent. Columns show the three stability regimes of the branching ratio $K$ for $h=7$, $n=3$, and flip probability $p=0.25$: (a),(d) $K=6$, fluctuations destabilize all but the deepest levels (only leaves are stable minima); (b),(e) $K=15$, the gap band clears zero after every node and all levels are stable; (c),(f) $K=60$, only the root remains stable, the relative fluctuations grow with depth toward the large-$K$ plateau $C_v^{2}\!\to\!q^{-2\ell}-1$.  Note the symmetric-log scales in (c),(f), reflecting the geometric growth of the root basin, $\sim(Kq^{n})^{h}$.}
    \label{fig:energy_diagram}
\end{figure}

~\cref{fig:energy_diagram} shows the mean energy landscape as we probe locations from the root to a leaf (following Hamming-distance geodesics from parent to child), passing through all intermediate nodes along the way. Each level is located at a minimum of the average energy, superimposed on an overall trend towards lower/higher energy as we move from the root to the leaves. Subfigure (a) shows that small $K$ leads to greater stability at the leaf, and in general moving down the hierarchy, while (b) shows the opposite effect at larger $K$. Correspondingly subfigures (d) and (e) show the average energy difference corresponding to a spin-flip, moving from the higher to lower prototypes. In (a) $K$ is too small for any but the leaf prototype to be stable which corresponds to the shaded variance of $\Delta E$ overlapping the $x$-axis (it's sign is uncertain) while in (b) the distribution is gapped from 0. (c) and (f) show that once $K$ becomes too large stability towards the leaves collapse entirely due to interference and only the prototypes near the root are stable. 

The local minima have basis of attraction which extend approximately halfway (in Hamming distance) to the immediate parent or child nodes. The sizes of these neighborhoods don't depend strongly on the parameters though the overall trend effects the degree to which points prefer to flow to lower or higher nodes on the tree. 

\mysection{Experiment on Fashion-MNIST}
We extend our analytics on the exactly solvable tree model to a real dataset which displays some features of hierarchical correlations, Fashion-MNIST \cite{xiao2017/online}. These images carry a natural multi-level structure with each image belonging to one of ten defined classes and those classes sharing relative similarity relations with each other. Additionally this dataset largely lacks translation or rotation symmetry which our simplified model is not designed for, but which a more general model should simultaneously manage \cite{kamb2024analytic}. We comment on this further in the discussion. The dense associate memory built from this dataset shows that the sharpness, $n$, still controls which level of the hierarchy is stable. Though the theoretical investigation primarily discusses a scaling of training examples, they also imply how for different $\ell$, the $(\mathcal{C}^v_\ell)^2$ may be adjusted with $n$ (see \cref{fig:setting3}(b)). 

We binarize (per-pixel dithering) the images and store them as memories in a dense associative memory with polynomial activation, initialize the dynamics at a single, ankle boot query image, and follow energy descent to a fixed point. We sweep the sharpness, and the number of images per class (which we analogize to the branching factor, $K$). This traces out the phase diagram of ~\cref{fig:phasediagram} where we color each point by Hamming distance between the recovered fixed point and the original query. We find four well-separated regimes corresponding to distinct levels of the latent hierarchy. For large $n$ or small $K$ the activation is sharp enough that the query is pinned to the query image (the leaf, red) and the network purely memorizes. As $n$ is lowered, the basin flows to an emergent boot prototype and then to the shoe prototype, neither of which are stored explicitly but appear as minima built from the interference of many examples as predicted by the tree model. For the smallest $n$ (left) the flow collapses onto the global root memory (purple star), the over-generalized fixed point. Near the lower triple point the recovered shoe prototype visibly borrows features from the boot prototype while remaining recognizably a shoe. Indeed the shoe prototype appears at a smaller $K$ than the boot prototype indicating that the basin it belongs to encompasses both types of footwear. This progression with decreasing $n$ mirrors the memorization-to-generalization crossover of the toy model, and the existence of stable intermediate phases is the real-data counterpart of the emergent ancestor memories. 
 
\begin{figure}
    \centering
    \includegraphics[width=0.99\linewidth]{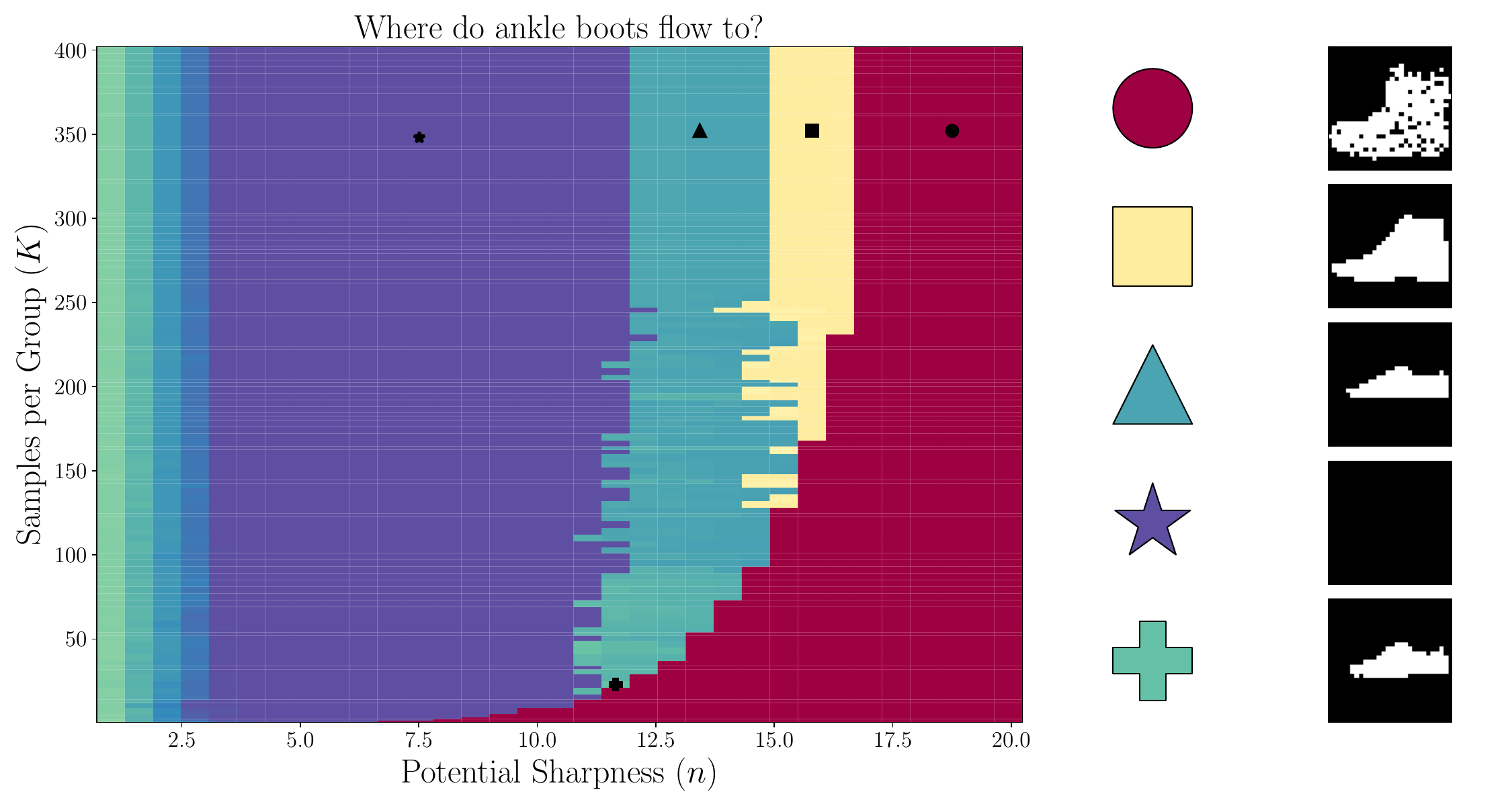}
    \caption{Phase diagram of the dense-associative memory dynamics on Fashion-MNIST initialized at an ankle boot image over a range of activation sharpnesses and group sizes. Color encodes the Hamming distance between the query and the recovered fixed point (red, circle is exactly zero). The five markers locate representative fixed points, shown as images on the right with matching shape and color. Four well-separated regions appear: the original image (leaf node), an emergent boot prototype, an emergent shoe prototype, and the blank root prototype. At the lower triple point (teal plus), the shoe prototype takes on some features of the boot prototype (e.g. a lift at the front of the shoe) while remaining largely consistent with the blue triangle shoe prototype.}
    \label{fig:phasediagram}
\end{figure}

\mysection{Discussion} 
In this work, we have explored a memorization-generalization transition that occurs in dense Hopfield networks, where the emergent memory patterns arise as a direct consequence of the underlying correlation structure in the data. When the data to be stored is correlated in hierarchical manner, then the ancestor patterns of the encoded memories themselves emerge as stable memories in the network. In this way, a type of generalization has occurred by way of prototype reconstruction, where encoded memories serve as ``examples'' for the network to learn a more general ``concept.'' We explore this idea both through a direct analytical signal-to-noise ratio calculation and by numerically probing the energy landscape. These analyses elucidate that the level of generalization is a function of the number of correlated patterns and the activation function, e.g. for some number of training examples below a number which depends on the strength of correlations, all ancestor prototypes have a finite probability to be stable minima. Finally, we qualitatively demonstrate some aspects of this phenomena in a model of real data, Fashion-MNIST.

This work may be extended in numerous directions. First, we have used only a simple one layer network. It would be interesting to explore what would change under more complicated architectures, such as multilayer Hopfield networks and even hierarchical Hopfield networks \cite{ramsauer2021hopfieldnetworksneed, krotov2021hierarchicalassociativememory}. In addition, given the enormous performance boost granted by the usage of self-attention and transformers in modern diffusion model architectures, it would be curious to ask how these types of layers could be used in these associative memories. For example, we performed numerical experiments on Fashion-MNIST, where the emergent memories yielded sensible patterns but in more complicated datasets, for example those with spatial symmetries, we found that that emergent memories became washed out and blurry. We hypothesize that usage of more complex architectures could generalize more complex hierarchically correlated data into discernable patterns. Indeed capturing latent structure was part of the philosophy behind the Boltzmann machine as an improvement on the original Hopfield model \cite{Ackley1985-up}. This work could be a starting point towards understanding the behavior of diffusion models as well as the attention mechanism in transformers.

We also remark on the similarity between our data model and the ultrametric structure of the Gibbs states in a full replica symmetry broken spin glass \cite{Mezard1987-tw, amit1985storing, Parga1986-rf}. In these f-RSB models, this ultrametric structure is an emergent property of the ruggedness of the energy landscape and the presence of exponentially many (in system size) stable local minima. What is interesting is that in the model we present, we also obtain an ultrametric structure of stable local minima, but this structure is "hard-coded" by the data model. We believe an analysis of the contrast of our model to that of a broken replica symmetry spin glass is an intriguing question for future work.

\mysection{Acknowledgements}
We acknowledge helpful discussions with Surya Ganguli, Federico Ghimenti, Mason Kamb, Dmitry Krotov, Aditya Mahadevan, and Akshat Pandey.
A.S. was supported in part by the Office of Naval Research Young Investigator Program (ONR YIP) under Award Number N00014-24-1-2098 and in part by a Packard Fellowship in Science and Engineering (PI: Vedika Khemani). A.S. also acknowledges support from the NSF graduate research fellowship and the ARCS Scholar award. A.C. acknowledges funding from the Simons Foundation (Award ID: 560571).

Claude was used to generate the cartoon in ~\cref{fig:cartoon1} and both Claude and ChatGPT  were used to assist with and proofread the analytical calculations done by the authors.

\let\oldaddcontentsline\addcontentsline%
\renewcommand{\addcontentsline}[3]{}%

\bibliographystyle{apsrev4-2}
\bibliography{main}

\begin{thebibliography}{36}%
\makeatletter
\providecommand \@ifxundefined [1]{%
 \@ifx{#1\undefined}
}%
\providecommand \@ifnum [1]{%
 \ifnum #1\expandafter \@firstoftwo
 \else \expandafter \@secondoftwo
 \fi
}%
\providecommand \@ifx [1]{%
 \ifx #1\expandafter \@firstoftwo
 \else \expandafter \@secondoftwo
 \fi
}%
\providecommand \natexlab [1]{#1}%
\providecommand \enquote  [1]{``#1''}%
\providecommand \bibnamefont  [1]{#1}%
\providecommand \bibfnamefont [1]{#1}%
\providecommand \citenamefont [1]{#1}%
\providecommand \href@noop [0]{\@secondoftwo}%
\providecommand \href [0]{\begingroup \@sanitize@url \@href}%
\providecommand \@href[1]{\@@startlink{#1}\@@href}%
\providecommand \@@href[1]{\endgroup#1\@@endlink}%
\providecommand \@sanitize@url [0]{\catcode `\\12\catcode `\$12\catcode `\&12\catcode `\#12\catcode `\^12\catcode `\_12\catcode `\%12\relax}%
\providecommand \@@startlink[1]{}%
\providecommand \@@endlink[0]{}%
\providecommand \url  [0]{\begingroup\@sanitize@url \@url }%
\providecommand \@url [1]{\endgroup\@href {#1}{\urlprefix }}%
\providecommand \urlprefix  [0]{URL }%
\providecommand \Eprint [0]{\href }%
\providecommand \doibase [0]{https://doi.org/}%
\providecommand \selectlanguage [0]{\@gobble}%
\providecommand \bibinfo  [0]{\@secondoftwo}%
\providecommand \bibfield  [0]{\@secondoftwo}%
\providecommand \translation [1]{[#1]}%
\providecommand \BibitemOpen [0]{}%
\providecommand \bibitemStop [0]{}%
\providecommand \bibitemNoStop [0]{.\EOS\space}%
\providecommand \EOS [0]{\spacefactor3000\relax}%
\providecommand \BibitemShut  [1]{\csname bibitem#1\endcsname}%
\let\auto@bib@innerbib\@empty
%</preamble>
\bibitem [{\citenamefont {Vaswani}\ \emph {et~al.}(2023)\citenamefont {Vaswani}, \citenamefont {Shazeer}, \citenamefont {Parmar}, \citenamefont {Uszkoreit}, \citenamefont {Jones}, \citenamefont {Gomez}, \citenamefont {Kaiser},\ and\ \citenamefont {Polosukhin}}]{vaswani2023attentionneed}%
  \BibitemOpen
  \bibfield  {author} {\bibinfo {author} {\bibfnamefont {A.}~\bibnamefont {Vaswani}}, \bibinfo {author} {\bibfnamefont {N.}~\bibnamefont {Shazeer}}, \bibinfo {author} {\bibfnamefont {N.}~\bibnamefont {Parmar}}, \bibinfo {author} {\bibfnamefont {J.}~\bibnamefont {Uszkoreit}}, \bibinfo {author} {\bibfnamefont {L.}~\bibnamefont {Jones}}, \bibinfo {author} {\bibfnamefont {A.~N.}\ \bibnamefont {Gomez}}, \bibinfo {author} {\bibfnamefont {L.}~\bibnamefont {Kaiser}},\ and\ \bibinfo {author} {\bibfnamefont {I.}~\bibnamefont {Polosukhin}},\ }\href {https://arxiv.org/abs/1706.03762} {\bibinfo {title} {Attention is all you need}} (\bibinfo {year} {2023}),\ \Eprint {https://arxiv.org/abs/1706.03762} {arXiv:1706.03762 [cs.CL]} \BibitemShut {NoStop}%
\bibitem [{\citenamefont {Devlin}\ \emph {et~al.}(2019)\citenamefont {Devlin}, \citenamefont {Chang}, \citenamefont {Lee},\ and\ \citenamefont {Toutanova}}]{devlin-etal-2019-bert}%
  \BibitemOpen
  \bibfield  {author} {\bibinfo {author} {\bibfnamefont {J.}~\bibnamefont {Devlin}}, \bibinfo {author} {\bibfnamefont {M.-W.}\ \bibnamefont {Chang}}, \bibinfo {author} {\bibfnamefont {K.}~\bibnamefont {Lee}},\ and\ \bibinfo {author} {\bibfnamefont {K.}~\bibnamefont {Toutanova}},\ }in\ \href {https://doi.org/10.18653/v1/N19-1423} {\emph {\bibinfo {booktitle} {Proceedings of the 2019 Conference of the North {A}merican Chapter of the Association for Computational Linguistics: Human Language Technologies, Volume 1 (Long and Short Papers)}}},\ \bibinfo {editor} {edited by\ \bibinfo {editor} {\bibfnamefont {J.}~\bibnamefont {Burstein}}, \bibinfo {editor} {\bibfnamefont {C.}~\bibnamefont {Doran}},\ and\ \bibinfo {editor} {\bibfnamefont {T.}~\bibnamefont {Solorio}}}\ (\bibinfo  {publisher} {Association for Computational Linguistics},\ \bibinfo {address} {Minneapolis, Minnesota},\ \bibinfo {year} {2019})\ pp.\ \bibinfo {pages} {4171--4186}\BibitemShut {NoStop}%
\bibitem [{\citenamefont {Dosovitskiy}\ \emph {et~al.}(2021)\citenamefont {Dosovitskiy}, \citenamefont {Beyer}, \citenamefont {Kolesnikov}, \citenamefont {Weissenborn}, \citenamefont {Zhai}, \citenamefont {Unterthiner}, \citenamefont {Dehghani}, \citenamefont {Minderer}, \citenamefont {Heigold}, \citenamefont {Gelly}, \citenamefont {Uszkoreit},\ and\ \citenamefont {Houlsby}}]{dosovitskiy2021imageworth16x16words}%
  \BibitemOpen
  \bibfield  {author} {\bibinfo {author} {\bibfnamefont {A.}~\bibnamefont {Dosovitskiy}}, \bibinfo {author} {\bibfnamefont {L.}~\bibnamefont {Beyer}}, \bibinfo {author} {\bibfnamefont {A.}~\bibnamefont {Kolesnikov}}, \bibinfo {author} {\bibfnamefont {D.}~\bibnamefont {Weissenborn}}, \bibinfo {author} {\bibfnamefont {X.}~\bibnamefont {Zhai}}, \bibinfo {author} {\bibfnamefont {T.}~\bibnamefont {Unterthiner}}, \bibinfo {author} {\bibfnamefont {M.}~\bibnamefont {Dehghani}}, \bibinfo {author} {\bibfnamefont {M.}~\bibnamefont {Minderer}}, \bibinfo {author} {\bibfnamefont {G.}~\bibnamefont {Heigold}}, \bibinfo {author} {\bibfnamefont {S.}~\bibnamefont {Gelly}}, \bibinfo {author} {\bibfnamefont {J.}~\bibnamefont {Uszkoreit}},\ and\ \bibinfo {author} {\bibfnamefont {N.}~\bibnamefont {Houlsby}},\ }\href {https://arxiv.org/abs/2010.11929} {\bibinfo {title} {An image is worth 16x16 words: Transformers for image recognition at scale}} (\bibinfo {year} {2021}),\ \Eprint {https://arxiv.org/abs/2010.11929} {arXiv:2010.11929
  [cs.CV]} \BibitemShut {NoStop}%
\bibitem [{\citenamefont {Peebles}\ and\ \citenamefont {Xie}(2023)}]{peebles2023scalablediffusionmodelstransformers}%
  \BibitemOpen
  \bibfield  {author} {\bibinfo {author} {\bibfnamefont {W.}~\bibnamefont {Peebles}}\ and\ \bibinfo {author} {\bibfnamefont {S.}~\bibnamefont {Xie}},\ }\href {https://arxiv.org/abs/2212.09748} {\bibinfo {title} {Scalable diffusion models with transformers}} (\bibinfo {year} {2023}),\ \Eprint {https://arxiv.org/abs/2212.09748} {arXiv:2212.09748 [cs.CV]} \BibitemShut {NoStop}%
\bibitem [{\citenamefont {Simon}\ \emph {et~al.}(2026)\citenamefont {Simon}, \citenamefont {Kunin}, \citenamefont {Atanasov}, \citenamefont {Boix-Adserà}, \citenamefont {Bordelon}, \citenamefont {Cohen}, \citenamefont {Ghosh}, \citenamefont {Guth}, \citenamefont {Jacot}, \citenamefont {Kamb}, \citenamefont {Karkada}, \citenamefont {Michaud}, \citenamefont {Ottlik},\ and\ \citenamefont {Turnbull}}]{simon2026scientifictheorydeeplearning}%
  \BibitemOpen
  \bibfield  {author} {\bibinfo {author} {\bibfnamefont {J.}~\bibnamefont {Simon}}, \bibinfo {author} {\bibfnamefont {D.}~\bibnamefont {Kunin}}, \bibinfo {author} {\bibfnamefont {A.}~\bibnamefont {Atanasov}}, \bibinfo {author} {\bibfnamefont {E.}~\bibnamefont {Boix-Adserà}}, \bibinfo {author} {\bibfnamefont {B.}~\bibnamefont {Bordelon}}, \bibinfo {author} {\bibfnamefont {J.}~\bibnamefont {Cohen}}, \bibinfo {author} {\bibfnamefont {N.}~\bibnamefont {Ghosh}}, \bibinfo {author} {\bibfnamefont {F.}~\bibnamefont {Guth}}, \bibinfo {author} {\bibfnamefont {A.}~\bibnamefont {Jacot}}, \bibinfo {author} {\bibfnamefont {M.}~\bibnamefont {Kamb}}, \bibinfo {author} {\bibfnamefont {D.}~\bibnamefont {Karkada}}, \bibinfo {author} {\bibfnamefont {E.~J.}\ \bibnamefont {Michaud}}, \bibinfo {author} {\bibfnamefont {B.}~\bibnamefont {Ottlik}},\ and\ \bibinfo {author} {\bibfnamefont {J.}~\bibnamefont {Turnbull}},\ }\href {https://arxiv.org/abs/2604.21691} {\bibinfo {title} {There will be a scientific theory of deep learning}}
  (\bibinfo {year} {2026}),\ \Eprint {https://arxiv.org/abs/2604.21691} {arXiv:2604.21691 [stat.ML]} \BibitemShut {NoStop}%
\bibitem [{\citenamefont {Sohl-Dickstein}\ \emph {et~al.}(2015)\citenamefont {Sohl-Dickstein}, \citenamefont {Weiss}, \citenamefont {Maheswaranathan},\ and\ \citenamefont {Ganguli}}]{sohldickstein2015deepunsupervisedlearningusing}%
  \BibitemOpen
  \bibfield  {author} {\bibinfo {author} {\bibfnamefont {J.}~\bibnamefont {Sohl-Dickstein}}, \bibinfo {author} {\bibfnamefont {E.~A.}\ \bibnamefont {Weiss}}, \bibinfo {author} {\bibfnamefont {N.}~\bibnamefont {Maheswaranathan}},\ and\ \bibinfo {author} {\bibfnamefont {S.}~\bibnamefont {Ganguli}},\ }\href {https://arxiv.org/abs/1503.03585} {\bibinfo {title} {Deep unsupervised learning using nonequilibrium thermodynamics}} (\bibinfo {year} {2015}),\ \Eprint {https://arxiv.org/abs/1503.03585} {arXiv:1503.03585 [cs.LG]} \BibitemShut {NoStop}%
\bibitem [{\citenamefont {Kamb}\ and\ \citenamefont {Ganguli}(2024)}]{kamb2024analytic}%
  \BibitemOpen
  \bibfield  {author} {\bibinfo {author} {\bibfnamefont {M.}~\bibnamefont {Kamb}}\ and\ \bibinfo {author} {\bibfnamefont {S.}~\bibnamefont {Ganguli}},\ }\href@noop {} {\bibfield  {journal} {\bibinfo  {journal} {arXiv preprint arXiv:2412.20292}\ } (\bibinfo {year} {2024})}\BibitemShut {NoStop}%
\bibitem [{\citenamefont {Hunt}\ \emph {et~al.}(2026)\citenamefont {Hunt}, \citenamefont {Kamb},\ and\ \citenamefont {Ganguli}}]{hunt2026exactinformationtheorygeneralization}%
  \BibitemOpen
  \bibfield  {author} {\bibinfo {author} {\bibfnamefont {H.}~\bibnamefont {Hunt}}, \bibinfo {author} {\bibfnamefont {M.}~\bibnamefont {Kamb}},\ and\ \bibinfo {author} {\bibfnamefont {S.}~\bibnamefont {Ganguli}},\ }\href {https://arxiv.org/abs/2607.08041} {\bibinfo {title} {An exact information theory of generalization phase transitions in bayesian diffusion models}} (\bibinfo {year} {2026}),\ \Eprint {https://arxiv.org/abs/2607.08041} {arXiv:2607.08041 [cs.LG]} \BibitemShut {NoStop}%
\bibitem [{\citenamefont {Pham}\ \emph {et~al.}(2026)\citenamefont {Pham}, \citenamefont {Raya}, \citenamefont {Negri}, \citenamefont {Zaki}, \citenamefont {Ambrogioni},\ and\ \citenamefont {Krotov}}]{pham2026memorizationgeneralizationemergencediffusion}%
  \BibitemOpen
  \bibfield  {author} {\bibinfo {author} {\bibfnamefont {B.}~\bibnamefont {Pham}}, \bibinfo {author} {\bibfnamefont {G.}~\bibnamefont {Raya}}, \bibinfo {author} {\bibfnamefont {M.}~\bibnamefont {Negri}}, \bibinfo {author} {\bibfnamefont {M.~J.}\ \bibnamefont {Zaki}}, \bibinfo {author} {\bibfnamefont {L.}~\bibnamefont {Ambrogioni}},\ and\ \bibinfo {author} {\bibfnamefont {D.}~\bibnamefont {Krotov}},\ }\href {https://arxiv.org/abs/2505.21777} {\bibinfo {title} {Memorization to generalization: Emergence of diffusion models from associative memory}} (\bibinfo {year} {2026}),\ \Eprint {https://arxiv.org/abs/2505.21777} {arXiv:2505.21777 [cs.LG]} \BibitemShut {NoStop}%
\bibitem [{\citenamefont {Niedoba}\ \emph {et~al.}(2025)\citenamefont {Niedoba}, \citenamefont {Zwartsenberg}, \citenamefont {Murphy},\ and\ \citenamefont {Wood}}]{niedoba2025mechanisticexplanationdiffusionmodel}%
  \BibitemOpen
  \bibfield  {author} {\bibinfo {author} {\bibfnamefont {M.}~\bibnamefont {Niedoba}}, \bibinfo {author} {\bibfnamefont {B.}~\bibnamefont {Zwartsenberg}}, \bibinfo {author} {\bibfnamefont {K.}~\bibnamefont {Murphy}},\ and\ \bibinfo {author} {\bibfnamefont {F.}~\bibnamefont {Wood}},\ }\href {https://arxiv.org/abs/2411.19339} {\bibinfo {title} {Towards a mechanistic explanation of diffusion model generalization}} (\bibinfo {year} {2025}),\ \Eprint {https://arxiv.org/abs/2411.19339} {arXiv:2411.19339 [cs.LG]} \BibitemShut {NoStop}%
\bibitem [{\citenamefont {Cui}\ \emph {et~al.}(2025)\citenamefont {Cui}, \citenamefont {Pehlevan},\ and\ \citenamefont {Lu}}]{cui2025solvablemodellearninggenerative}%
  \BibitemOpen
  \bibfield  {author} {\bibinfo {author} {\bibfnamefont {H.}~\bibnamefont {Cui}}, \bibinfo {author} {\bibfnamefont {C.}~\bibnamefont {Pehlevan}},\ and\ \bibinfo {author} {\bibfnamefont {Y.~M.}\ \bibnamefont {Lu}},\ }\href {https://arxiv.org/abs/2501.03937} {\bibinfo {title} {A solvable model of learning generative diffusion: theory and insights}} (\bibinfo {year} {2025}),\ \Eprint {https://arxiv.org/abs/2501.03937} {arXiv:2501.03937 [cs.LG]} \BibitemShut {NoStop}%
\bibitem [{\citenamefont {Hopfield}(1982)}]{doi:10.1073/pnas.79.8.2554}%
  \BibitemOpen
  \bibfield  {author} {\bibinfo {author} {\bibfnamefont {J.~J.}\ \bibnamefont {Hopfield}},\ }\href {https://doi.org/10.1073/pnas.79.8.2554} {\bibfield  {journal} {\bibinfo  {journal} {Proceedings of the National Academy of Sciences}\ }\textbf {\bibinfo {volume} {79}},\ \bibinfo {pages} {2554} (\bibinfo {year} {1982})},\ \Eprint {https://arxiv.org/abs/https://www.pnas.org/doi/pdf/10.1073/pnas.79.8.2554} {https://www.pnas.org/doi/pdf/10.1073/pnas.79.8.2554} \BibitemShut {NoStop}%
\bibitem [{\citenamefont {Krotov}\ \emph {et~al.}(2025)\citenamefont {Krotov}, \citenamefont {Hoover}, \citenamefont {Ram},\ and\ \citenamefont {Pham}}]{krotov2025modernmethodsassociativememory}%
  \BibitemOpen
  \bibfield  {author} {\bibinfo {author} {\bibfnamefont {D.}~\bibnamefont {Krotov}}, \bibinfo {author} {\bibfnamefont {B.}~\bibnamefont {Hoover}}, \bibinfo {author} {\bibfnamefont {P.}~\bibnamefont {Ram}},\ and\ \bibinfo {author} {\bibfnamefont {B.}~\bibnamefont {Pham}},\ }\href {https://arxiv.org/abs/2507.06211} {\bibinfo {title} {Modern methods in associative memory}} (\bibinfo {year} {2025}),\ \Eprint {https://arxiv.org/abs/2507.06211} {arXiv:2507.06211 [cs.LG]} \BibitemShut {NoStop}%
\bibitem [{\citenamefont {Ramsauer}\ \emph {et~al.}(2021)\citenamefont {Ramsauer}, \citenamefont {Schäfl}, \citenamefont {Lehner}, \citenamefont {Seidl}, \citenamefont {Widrich}, \citenamefont {Adler}, \citenamefont {Gruber}, \citenamefont {Holzleitner}, \citenamefont {Pavlović}, \citenamefont {Sandve}, \citenamefont {Greiff}, \citenamefont {Kreil}, \citenamefont {Kopp}, \citenamefont {Klambauer}, \citenamefont {Brandstetter},\ and\ \citenamefont {Hochreiter}}]{ramsauer2021hopfieldnetworksneed}%
  \BibitemOpen
  \bibfield  {author} {\bibinfo {author} {\bibfnamefont {H.}~\bibnamefont {Ramsauer}}, \bibinfo {author} {\bibfnamefont {B.}~\bibnamefont {Schäfl}}, \bibinfo {author} {\bibfnamefont {J.}~\bibnamefont {Lehner}}, \bibinfo {author} {\bibfnamefont {P.}~\bibnamefont {Seidl}}, \bibinfo {author} {\bibfnamefont {M.}~\bibnamefont {Widrich}}, \bibinfo {author} {\bibfnamefont {T.}~\bibnamefont {Adler}}, \bibinfo {author} {\bibfnamefont {L.}~\bibnamefont {Gruber}}, \bibinfo {author} {\bibfnamefont {M.}~\bibnamefont {Holzleitner}}, \bibinfo {author} {\bibfnamefont {M.}~\bibnamefont {Pavlović}}, \bibinfo {author} {\bibfnamefont {G.~K.}\ \bibnamefont {Sandve}}, \bibinfo {author} {\bibfnamefont {V.}~\bibnamefont {Greiff}}, \bibinfo {author} {\bibfnamefont {D.}~\bibnamefont {Kreil}}, \bibinfo {author} {\bibfnamefont {M.}~\bibnamefont {Kopp}}, \bibinfo {author} {\bibfnamefont {G.}~\bibnamefont {Klambauer}}, \bibinfo {author} {\bibfnamefont {J.}~\bibnamefont {Brandstetter}},\ and\ \bibinfo {author} {\bibfnamefont
  {S.}~\bibnamefont {Hochreiter}},\ }\href {https://arxiv.org/abs/2008.02217} {\bibinfo {title} {Hopfield networks is all you need}} (\bibinfo {year} {2021}),\ \Eprint {https://arxiv.org/abs/2008.02217} {arXiv:2008.02217 [cs.NE]} \BibitemShut {NoStop}%
\bibitem [{\citenamefont {Gaudio}\ \emph {et~al.}(2026)\citenamefont {Gaudio}, \citenamefont {Ghimenti},\ and\ \citenamefont {Ganguli}}]{delgaudio2026shorttermplasticityrecallsforgotten}%
  \BibitemOpen
  \bibfield  {author} {\bibinfo {author} {\bibfnamefont {M.~D.}\ \bibnamefont {Gaudio}}, \bibinfo {author} {\bibfnamefont {F.}~\bibnamefont {Ghimenti}},\ and\ \bibinfo {author} {\bibfnamefont {S.}~\bibnamefont {Ganguli}},\ }\href {https://arxiv.org/abs/2511.22848} {\bibinfo {title} {Short-term plasticity recalls forgotten memories through a trampoline mechanism}} (\bibinfo {year} {2026}),\ \Eprint {https://arxiv.org/abs/2511.22848} {arXiv:2511.22848 [q-bio.NC]} \BibitemShut {NoStop}%
\bibitem [{\citenamefont {Lufkin}\ \emph {et~al.}(2026)\citenamefont {Lufkin}, \citenamefont {Figliolia}, \citenamefont {Millidge},\ and\ \citenamefont {Krishnamurthy}}]{lufkin2026hybridassociativememories}%
  \BibitemOpen
  \bibfield  {author} {\bibinfo {author} {\bibfnamefont {L.}~\bibnamefont {Lufkin}}, \bibinfo {author} {\bibfnamefont {T.}~\bibnamefont {Figliolia}}, \bibinfo {author} {\bibfnamefont {B.}~\bibnamefont {Millidge}},\ and\ \bibinfo {author} {\bibfnamefont {K.}~\bibnamefont {Krishnamurthy}},\ }\href {https://arxiv.org/abs/2603.22325} {\bibinfo {title} {Hybrid associative memories}} (\bibinfo {year} {2026}),\ \Eprint {https://arxiv.org/abs/2603.22325} {arXiv:2603.22325 [cs.LG]} \BibitemShut {NoStop}%
\bibitem [{\citenamefont {Steinberg}\ and\ \citenamefont {Sompolinsky}(2022)}]{Steinberg2022-jd}%
  \BibitemOpen
  \bibfield  {author} {\bibinfo {author} {\bibfnamefont {J.}~\bibnamefont {Steinberg}}\ and\ \bibinfo {author} {\bibfnamefont {H.}~\bibnamefont {Sompolinsky}},\ }\href@noop {} {\bibfield  {journal} {\bibinfo  {journal} {Sci. Rep.}\ }\textbf {\bibinfo {volume} {12}},\ \bibinfo {pages} {21808} (\bibinfo {year} {2022})}\BibitemShut {NoStop}%
\bibitem [{\citenamefont {Krotov}\ and\ \citenamefont {Hopfield}(2016)}]{krotov2016dense}%
  \BibitemOpen
  \bibfield  {author} {\bibinfo {author} {\bibfnamefont {D.}~\bibnamefont {Krotov}}\ and\ \bibinfo {author} {\bibfnamefont {J.~J.}\ \bibnamefont {Hopfield}},\ }\href@noop {} {\bibfield  {journal} {\bibinfo  {journal} {Advances in neural information processing systems}\ }\textbf {\bibinfo {volume} {29}} (\bibinfo {year} {2016})}\BibitemShut {NoStop}%
\bibitem [{\citenamefont {Lucibello}\ and\ \citenamefont {M\'ezard}(2024)}]{PhysRevLett.132.077301}%
  \BibitemOpen
  \bibfield  {author} {\bibinfo {author} {\bibfnamefont {C.}~\bibnamefont {Lucibello}}\ and\ \bibinfo {author} {\bibfnamefont {M.}~\bibnamefont {M\'ezard}},\ }\href {https://doi.org/10.1103/PhysRevLett.132.077301} {\bibfield  {journal} {\bibinfo  {journal} {Phys. Rev. Lett.}\ }\textbf {\bibinfo {volume} {132}},\ \bibinfo {pages} {077301} (\bibinfo {year} {2024})}\BibitemShut {NoStop}%
\bibitem [{\citenamefont {Kafraj}\ \emph {et~al.}(2026)\citenamefont {Kafraj}, \citenamefont {Krotov},\ and\ \citenamefont {Latham}}]{kafraj2026biologicallyplausibledenseassociative}%
  \BibitemOpen
  \bibfield  {author} {\bibinfo {author} {\bibfnamefont {M.~S.}\ \bibnamefont {Kafraj}}, \bibinfo {author} {\bibfnamefont {D.}~\bibnamefont {Krotov}},\ and\ \bibinfo {author} {\bibfnamefont {P.~E.}\ \bibnamefont {Latham}},\ }\href {https://arxiv.org/abs/2601.00984} {\bibinfo {title} {A biologically plausible dense associative memory with exponential capacity}} (\bibinfo {year} {2026}),\ \Eprint {https://arxiv.org/abs/2601.00984} {arXiv:2601.00984 [q-bio.NC]} \BibitemShut {NoStop}%
\bibitem [{\citenamefont {Ambrogioni}(2023)}]{ambrogioni2023search}%
  \BibitemOpen
  \bibfield  {author} {\bibinfo {author} {\bibfnamefont {L.}~\bibnamefont {Ambrogioni}},\ }\href@noop {} {\bibfield  {journal} {\bibinfo  {journal} {arXiv preprint arXiv:2309.17290}\ } (\bibinfo {year} {2023})}\BibitemShut {NoStop}%
\bibitem [{\citenamefont {Niu}\ \emph {et~al.}(2024)\citenamefont {Niu}, \citenamefont {Bai}, \citenamefont {Deng},\ and\ \citenamefont {Han}}]{niu2024scalinglawsunderstandingtransformer}%
  \BibitemOpen
  \bibfield  {author} {\bibinfo {author} {\bibfnamefont {X.}~\bibnamefont {Niu}}, \bibinfo {author} {\bibfnamefont {B.}~\bibnamefont {Bai}}, \bibinfo {author} {\bibfnamefont {L.}~\bibnamefont {Deng}},\ and\ \bibinfo {author} {\bibfnamefont {W.}~\bibnamefont {Han}},\ }\href {https://arxiv.org/abs/2405.08707} {\bibinfo {title} {Beyond scaling laws: Understanding transformer performance with associative memory}} (\bibinfo {year} {2024}),\ \Eprint {https://arxiv.org/abs/2405.08707} {arXiv:2405.08707 [cs.LG]} \BibitemShut {NoStop}%
\bibitem [{\citenamefont {Hoover}\ \emph {et~al.}(2026)\citenamefont {Hoover}, \citenamefont {Shi}, \citenamefont {Balasubramanian}, \citenamefont {Krotov},\ and\ \citenamefont {Ram}}]{hoover2026denseassociativememoryepanechnikov}%
  \BibitemOpen
  \bibfield  {author} {\bibinfo {author} {\bibfnamefont {B.}~\bibnamefont {Hoover}}, \bibinfo {author} {\bibfnamefont {Z.}~\bibnamefont {Shi}}, \bibinfo {author} {\bibfnamefont {K.}~\bibnamefont {Balasubramanian}}, \bibinfo {author} {\bibfnamefont {D.}~\bibnamefont {Krotov}},\ and\ \bibinfo {author} {\bibfnamefont {P.}~\bibnamefont {Ram}},\ }\href {https://arxiv.org/abs/2506.10801} {\bibinfo {title} {Dense associative memory with epanechnikov energy}} (\bibinfo {year} {2026}),\ \Eprint {https://arxiv.org/abs/2506.10801} {arXiv:2506.10801 [cs.LG]} \BibitemShut {NoStop}%
\bibitem [{\citenamefont {Dotsenko}(1986)}]{DOTSENKO1986410}%
  \BibitemOpen
  \bibfield  {author} {\bibinfo {author} {\bibfnamefont {V.~S.}\ \bibnamefont {Dotsenko}},\ }\href {https://doi.org/https://doi.org/10.1016/0378-4371(86)90248-7} {\bibfield  {journal} {\bibinfo  {journal} {Physica A: Statistical Mechanics and its Applications}\ }\textbf {\bibinfo {volume} {140}},\ \bibinfo {pages} {410} (\bibinfo {year} {1986})}\BibitemShut {NoStop}%
\bibitem [{\citenamefont {Engel}(1990)}]{engelJPA}%
  \BibitemOpen
  \bibfield  {author} {\bibinfo {author} {\bibfnamefont {A.}~\bibnamefont {Engel}},\ }\bibfield  {booktitle} {\emph {\bibinfo {booktitle} {Journal of Physics A: Mathematical and General}},\ }\href {https://doi.org/10.1088/0305-4470/23/12/034} {\ \textbf {\bibinfo {volume} {23}},\ \bibinfo {pages} {2587} (\bibinfo {year} {1990})}\BibitemShut {NoStop}%
\bibitem [{\citenamefont {Agliari}\ \emph {et~al.}(2013)\citenamefont {Agliari}, \citenamefont {Barra}, \citenamefont {De~Antoni},\ and\ \citenamefont {Galluzzi}}]{Agliari:2013wo}%
  \BibitemOpen
  \bibfield  {author} {\bibinfo {author} {\bibfnamefont {E.}~\bibnamefont {Agliari}}, \bibinfo {author} {\bibfnamefont {A.}~\bibnamefont {Barra}}, \bibinfo {author} {\bibfnamefont {A.}~\bibnamefont {De~Antoni}},\ and\ \bibinfo {author} {\bibfnamefont {A.}~\bibnamefont {Galluzzi}},\ }\href {https://doi.org/10.1016/j.neunet.2012.11.010} {\bibfield  {journal} {\bibinfo  {journal} {Neural Netw}\ }\textbf {\bibinfo {volume} {38}},\ \bibinfo {pages} {52} (\bibinfo {year} {2013})}\BibitemShut {NoStop}%
\bibitem [{\citenamefont {Sargolzaei}\ and\ \citenamefont {Rueda}(2025)}]{sargolzaei2025hierarchical}%
  \BibitemOpen
  \bibfield  {author} {\bibinfo {author} {\bibfnamefont {S.}~\bibnamefont {Sargolzaei}}\ and\ \bibinfo {author} {\bibfnamefont {L.}~\bibnamefont {Rueda}},\ }in\ \href@noop {} {\emph {\bibinfo {booktitle} {New Frontiers in Associative Memory Workshop at the International Conference on Learning Representations (ICLR)}}}\ (\bibinfo {year} {2025})\BibitemShut {NoStop}%
\bibitem [{\citenamefont {Krotov}(2021)}]{krotov2021hierarchicalassociativememory}%
  \BibitemOpen
  \bibfield  {author} {\bibinfo {author} {\bibfnamefont {D.}~\bibnamefont {Krotov}},\ }\href {https://arxiv.org/abs/2107.06446} {\bibinfo {title} {Hierarchical associative memory}} (\bibinfo {year} {2021}),\ \Eprint {https://arxiv.org/abs/2107.06446} {arXiv:2107.06446 [cs.NE]} \BibitemShut {NoStop}%
\bibitem [{\citenamefont {Agliari}\ \emph {et~al.}(2023)\citenamefont {Agliari}, \citenamefont {Albanese}, \citenamefont {Alemanno}, \citenamefont {Alessandrelli}, \citenamefont {Barra}, \citenamefont {Giannotti}, \citenamefont {Lotito},\ and\ \citenamefont {Pedreschi}}]{Agliari2023-an}%
  \BibitemOpen
  \bibfield  {author} {\bibinfo {author} {\bibfnamefont {E.}~\bibnamefont {Agliari}}, \bibinfo {author} {\bibfnamefont {L.}~\bibnamefont {Albanese}}, \bibinfo {author} {\bibfnamefont {F.}~\bibnamefont {Alemanno}}, \bibinfo {author} {\bibfnamefont {A.}~\bibnamefont {Alessandrelli}}, \bibinfo {author} {\bibfnamefont {A.}~\bibnamefont {Barra}}, \bibinfo {author} {\bibfnamefont {F.}~\bibnamefont {Giannotti}}, \bibinfo {author} {\bibfnamefont {D.}~\bibnamefont {Lotito}},\ and\ \bibinfo {author} {\bibfnamefont {D.}~\bibnamefont {Pedreschi}},\ }\href@noop {} {\bibfield  {journal} {\bibinfo  {journal} {Physica A}\ }\textbf {\bibinfo {volume} {626}},\ \bibinfo {pages} {129076} (\bibinfo {year} {2023})}\BibitemShut {NoStop}%
\bibitem [{\citenamefont {Agliari}\ \emph {et~al.}(2021)\citenamefont {Agliari}, \citenamefont {Alemanno}, \citenamefont {Barra},\ and\ \citenamefont {Marzo}}]{agliari2021emergenceconceptshallowneural}%
  \BibitemOpen
  \bibfield  {author} {\bibinfo {author} {\bibfnamefont {E.}~\bibnamefont {Agliari}}, \bibinfo {author} {\bibfnamefont {F.}~\bibnamefont {Alemanno}}, \bibinfo {author} {\bibfnamefont {A.}~\bibnamefont {Barra}},\ and\ \bibinfo {author} {\bibfnamefont {G.~D.}\ \bibnamefont {Marzo}},\ }\href {https://arxiv.org/abs/2109.00454} {\bibinfo {title} {The emergence of a concept in shallow neural networks}} (\bibinfo {year} {2021}),\ \Eprint {https://arxiv.org/abs/2109.00454} {arXiv:2109.00454 [cond-mat.dis-nn]} \BibitemShut {NoStop}%
\bibitem [{\citenamefont {Geszti}(1990)}]{Geszti1990-yk}%
  \BibitemOpen
  \bibfield  {author} {\bibinfo {author} {\bibfnamefont {T.}~\bibnamefont {Geszti}},\ }\href@noop {} {\emph {\bibinfo {title} {Physical models of neural networks}}}\ (\bibinfo  {publisher} {World Scientific Publishing},\ \bibinfo {address} {Singapore, Singapore},\ \bibinfo {year} {1990})\BibitemShut {NoStop}%
\bibitem [{\citenamefont {Xiao}\ \emph {et~al.}(2017)\citenamefont {Xiao}, \citenamefont {Rasul},\ and\ \citenamefont {Vollgraf}}]{xiao2017/online}%
  \BibitemOpen
  \bibfield  {author} {\bibinfo {author} {\bibfnamefont {H.}~\bibnamefont {Xiao}}, \bibinfo {author} {\bibfnamefont {K.}~\bibnamefont {Rasul}},\ and\ \bibinfo {author} {\bibfnamefont {R.}~\bibnamefont {Vollgraf}},\ }\href@noop {} {\bibinfo {title} {Fashion-mnist: a novel image dataset for benchmarking machine learning algorithms}} (\bibinfo {year} {2017}),\ \Eprint {https://arxiv.org/abs/cs.LG/1708.07747} {cs.LG/1708.07747} \BibitemShut {NoStop}%
\bibitem [{\citenamefont {Ackley}\ \emph {et~al.}(1985)\citenamefont {Ackley}, \citenamefont {Hinton},\ and\ \citenamefont {Sejnowski}}]{Ackley1985-up}%
  \BibitemOpen
  \bibfield  {author} {\bibinfo {author} {\bibfnamefont {D.}~\bibnamefont {Ackley}}, \bibinfo {author} {\bibfnamefont {G.}~\bibnamefont {Hinton}},\ and\ \bibinfo {author} {\bibfnamefont {T.}~\bibnamefont {Sejnowski}},\ }\href@noop {} {\bibfield  {journal} {\bibinfo  {journal} {Cogn. Sci.}\ }\textbf {\bibinfo {volume} {9}},\ \bibinfo {pages} {147} (\bibinfo {year} {1985})}\BibitemShut {NoStop}%
\bibitem [{\citenamefont {Mezard}\ \emph {et~al.}(1987)\citenamefont {Mezard}, \citenamefont {Parisi},\ and\ \citenamefont {Virasoro}}]{Mezard1987-tw}%
  \BibitemOpen
  \bibfield  {author} {\bibinfo {author} {\bibfnamefont {M.}~\bibnamefont {Mezard}}, \bibinfo {author} {\bibfnamefont {G.}~\bibnamefont {Parisi}},\ and\ \bibinfo {author} {\bibfnamefont {M.~A.}\ \bibnamefont {Virasoro}},\ }\href@noop {} {\emph {\bibinfo {title} {Spin glass theory and beyond: An introduction to the replica method and its applications}}},\ World Scientific Lecture Notes In Physics\ (\bibinfo  {publisher} {World Scientific Publishing},\ \bibinfo {address} {Singapore, Singapore},\ \bibinfo {year} {1987})\BibitemShut {NoStop}%
\bibitem [{\citenamefont {Amit}\ \emph {et~al.}(1985)\citenamefont {Amit}, \citenamefont {Gutfreund},\ and\ \citenamefont {Sompolinsky}}]{amit1985storing}%
  \BibitemOpen
  \bibfield  {author} {\bibinfo {author} {\bibfnamefont {D.~J.}\ \bibnamefont {Amit}}, \bibinfo {author} {\bibfnamefont {H.}~\bibnamefont {Gutfreund}},\ and\ \bibinfo {author} {\bibfnamefont {H.}~\bibnamefont {Sompolinsky}},\ }\href@noop {} {\bibfield  {journal} {\bibinfo  {journal} {Physical Review Letters}\ }\textbf {\bibinfo {volume} {55}},\ \bibinfo {pages} {1530} (\bibinfo {year} {1985})}\BibitemShut {NoStop}%
\bibitem [{\citenamefont {Parga}\ and\ \citenamefont {Virasoro}(1986)}]{Parga1986-rf}%
  \BibitemOpen
  \bibfield  {author} {\bibinfo {author} {\bibfnamefont {N.}~\bibnamefont {Parga}}\ and\ \bibinfo {author} {\bibfnamefont {M.~A.}\ \bibnamefont {Virasoro}},\ }in\ \href@noop {} {\emph {\bibinfo {booktitle} {World Scientific Lecture Notes in Physics}}}\ (\bibinfo  {publisher} {WORLD SCIENTIFIC},\ \bibinfo {year} {1986})\ pp.\ \bibinfo {pages} {436--443}\BibitemShut {NoStop}%
\end{thebibliography}%

\clearpage

\onecolumngrid

\setcounter{equation}{0}
\setcounter{figure}{0}
\setcounter{table}{0}
\setcounter{page}{1}
\makeatletter
\renewcommand{\thesection}{S\arabic{section}}
\renewcommand{\theequation}{S\arabic{equation}}
\renewcommand{\thefigure}{S\arabic{figure}}

\begin{center}
\textbf{\large Supplemental Material for ``Hierarchical Prototype Emergence in Modern Hopfield Models'' } \\~\\
Aditya Cowsik and Adithya Sriram \\
\textit{
 Department of Physics, Stanford University, Stanford, CA 94305, USA
}

\end{center}

\section{Moment Formulas and Diagrammatic Formalism}

In this section we will derive the analytical results for the energy gap statistics. To begin, let us first formally define the model described in the main text.
\begin{defn} [Dense Hopfield Model with Hierarchical Memories] \label{def: model formal def}
    Let $\mathbf{g}^0 \in \{-1,1\}^N$ be an iid Bernoulli(1/2) random vector. From $\mathbf{g}^0$, generate $K$ new datapoints by sampling random noise vectors $\mathbf{e} \in \{-1,1\}^N$, where each element $e_i$ is i.i.d. according to a Bernoulli distribution with parameter $p$, i.e. $\mathbb{P}(e_i = -1) = p$. We always take $0 < p < 1/2$. The new datapoints are generated by element-wise multiplication with the roots, $\mathbf{g}^1 = \mathbf{g}^0 \circ \mathbf{e}$.

    To form the next layer, generate $K-1$ new datapoints for each $\mathbf{g}^1$ by sampling random noise vectors according to the same distribution. This process is repeated until $h$ layers are generated. The patterns in the final layer, the leaf patterns, are denoted as $\boldsymbol{\xi}^\mu$ and there are overall $K(K-1)^{h-1}$ leaf patterns. The energy function for the dense Hopfield model is given by
    \begin{align}
        E(\boldsymbol{\sigma}) = -\sum_{\mu = 1}^M \left( \boldsymbol{\xi}^\mu \cdot \boldsymbol{\sigma}\right)^n
    \end{align}
    where $n$ is an even integer.
\end{defn}

We seek to calculate the statistics of the energy gap:

\begin{align} \label{eq: mth moment}
    \langle (\Delta E )^m \rangle 
    &= \left \langle  \left[ 2\sum_{\mu = 1}^M \sum_{k\  \rm odd}^n {n \choose k} (\xi_i^\mu \sigma_i) \left(\sum_{j \neq i} \xi^\mu_j \sigma_j \right)^{n-k}  \right]^m \right\rangle 
\end{align}

We will utilize the following result that allows us to more easily calculate the $m$th moment of the energy gap.

\begin{theo} [Approximate $m$th Moment of Energy Gap] \label{theorem: approximate gap}

    Let $\langle (\Delta \tilde{E} )^m \rangle$ be
    \begin{align} \label{eq: expanded gap}
        \langle (\Delta \tilde{E} )^m \rangle = 2^m n^m\sum_{\mu_1,\ldots,\mu_m = 1}^M \;\left\langle \prod_{a=1}^m \, (\xi_i^{\mu_a} \sigma_i)\right\rangle \left(\prod_a^m N^{n-1}q^{(n-1) d(\boldsymbol{\xi}^{\mu_a},\boldsymbol{\sigma})} \right)
    \end{align}
    For $q^{2h} = \omega(mn/\sqrt{N})$, the $m$th moment of the energy gap (\cref{eq: mth moment}) is approximated by 
    \begin{align}  \label{eq: approximate energy gap}
           \langle (\Delta E )^m \rangle\  \leq\  \langle (\Delta \tilde{E} )^m \rangle \cdot (1+o(1)),
    \end{align}

\end{theo}

We will prove this theorem upon introduction of two helper lemmas, the first of which allows us to restrict to $k = 1$ in \cref{eq: mth moment}, and the second of which which controls expectation values of terms of the form $\left(\sum_{j \neq i} \xi^\mu_j \sigma_j \right)^{n-k}$.

First, we expand out the part of the expectation value argument raised to the $m$:
\begin{align} \label{eq:expanded moment}
    &= 2^m \sum_{\mu_1,\ldots,\mu_m = 1}^M \;\sum_{\substack{k_1,\ldots,k_m \\ \rm all\ odd}}^n \;\left \langle\prod_{a=1}^m {n \choose k_a}\, (\xi_i^{\mu_a} \sigma_i) \left(\sum_{j \neq i} \xi^{\mu_a}_j \sigma_j \right)^{n-k_a} \right\rangle \\
    &= 2^m \sum_{\mu_1,\ldots,\mu_m = 1}^M \;\sum_{\substack{k_1,\ldots,k_m \\ \rm all\ odd}}^n \;T_\mathbf{k},
\end{align}
where 
\begin{align} \label{eq:T}
         T_\mathbf{k} &= \left \langle \prod_{a=1}^m {n \choose k_a}\, (\xi_i^{\mu_a} \sigma_i) \left(\sum_{j \neq i} \xi^{\mu_a}_j \sigma_j \right)^{n-k_a} \right\rangle.
    \end{align}

We may approximate this formula by only considering the term of the sum where all $k_a = 1$, i.e. only using $T_{(1,1,\dots, 1)}$. For a particular parameter regime, this leads to a controlled relative error on the moments of the energy gap. 

\begin{lem} [All $k_a = 1$ approximation] \label{lemma:all k 1}
    For $q^{2h} = \omega(mn / \sqrt{N})$, the $m$th moment of the energy gap (\cref{eq: mth moment}) is approximated by
    \begin{align}
        2^m \sum_{\mu_1,\ldots,\mu_m = 1}^M \; T_{(1,1,\dots 1)}\  \leq\  \langle (\Delta E)^m \rangle \leq\   2^m \sum_{\mu_1,\ldots,\mu_m = 1}^M \; T_{(1,1,\dots 1)}\cdot \left( 1 +  \varepsilon \right)
    \end{align}
    where $T_{(1,1,\dots,1)}$ is given by \cref{eq:T} and setting $\mathbf{k} = (1,1,\dots,1)$ and $\varepsilon = o(1)$ is a small relative error vanishing with $N$.
\end{lem}

The proof of this is given in \cref{sec: lemma2 proof}. The proof is not complicated but is rather tedious. \\

Next, let us write out $T_{(1,1,\dots,1)}$.
\begin{align}
    T_{(1,1\dots,1)} = n^m  \left \langle\prod_{a=1}^m  (\xi_i^{\mu_a} \sigma_i) \right\rangle \left\langle \prod_{a=1}^m \left(\sum_{j \neq i} \xi^{\mu_a}_j \sigma_j \right)^{n-1} \right\rangle
\end{align}

The first factor $ \left \langle\prod_{a=1}^m  (\xi_i^{\mu_a} \sigma_i) \right\rangle$ is simple to deal with. Recall that $q = (1-2p)$ so the overlap between any two nodes at a path-distance $d$ is
\begin{align}
    \langle \xi_i^{A} \xi_j^{B} \rangle = \delta_{ij} q^{d}.
\end{align}
The variable $d$ counts the number of edges on the unique simple path, also the shortest path, between $A$ and $B$. Thus, $\langle \xi_j^{\mu_a}\sigma_j \rangle$ is a simple function of the distance between $\boldsymbol{\xi}^{\mu_a}$ and $\boldsymbol{\sigma}$, i,e.
\begin{align}
    \langle \xi_j^{\mu_a}\sigma_j \rangle = q^{d(\boldsymbol{\xi}^{\mu_a}, \boldsymbol{\sigma})}.
\end{align}

Thus, to determine the overall term $\left\langle \prod_a^m (\xi_i^{\mu_a} \sigma_i) \right\rangle$, we simply count all the links which appear an odd number of times along the paths from $\boldsymbol{\sigma}$ to each of the $\boldsymbol{\xi}^{\mu_a}$. 

The second type of term, $\left\langle \prod_a^m \left( \sum_{j \neq i} \xi_j^{\mu_a} \sigma_j \right)^{n-1} \right\rangle$ is more complicated. In theory this requires calculating higher moment correlators between patterns. However, this term is well approximated by instead taking powers of the mean of the argument.

\begin{lem} [Bulk Propagator] \label{lemma: bulk propagator}
    The bulk propagator $\left\langle \prod_a^m \left( \sum_{j \neq i} \xi_j^{\mu_a} \sigma_j \right)^{n-1} \right\rangle$ is well approximated by its mean when $q^{2h} = \omega(mn/\sqrt{N})$, i.e.
    \begin{align}
        \left[\frac{(N-1)!}{(N-1- m(n-1))!} \prod_a^m q^{(n-1) d(\boldsymbol{\xi}^{\mu_a},\boldsymbol{\sigma})} \right] \leq \left\langle \prod_a^m \left( \sum_{j \neq i} \xi_j^{\mu_a} \sigma_j \right)^{n-1} \right\rangle \leq  \left[\prod_a^m N^{n-1}q^{(n-1) d(\boldsymbol{\xi}^{\mu_a},\boldsymbol{\sigma})} \right](1+\bar{\epsilon}),
    \end{align}
    where $\bar{\epsilon} = o(1)$.
\end{lem}

We present the proof of this lemma in \cref{sec: proof of bulk propagator}. Much of it involves the same series of steps as the proof of \cref{sec: lemma2 proof}.

Putting these two together, we may prove \cref{theorem: approximate gap}.

\begin{proof} [Proof of \cref{theorem: approximate gap}]
    Unfolding the $m$th moment of the gap, we obtain 

    \begin{align}
        \langle (\Delta E)^m \rangle &= 2^m \sum_{\mu_1,\ldots,\mu_m = 1}^M \;\sum_{\substack{k_1,\ldots,k_m \\ \rm all\ odd}}^n \;\left \langle\prod_{a=1}^m {n \choose k_a}\, (\xi_i^{\mu_a} \sigma_i) \left(\sum_{j \neq i} \xi^{\mu_a}_j \sigma_j \right)^{n-k_a} \right\rangle \\
    &\leq  2^m n^m \sum_{\mu_1,\ldots,\mu_m = 1}^M \;  \left \langle\prod_{a=1}^m  (\xi_i^{\mu_a} \sigma_i) \right\rangle \left\langle \prod_{a=1}^m \left(\sum_{j \neq i} \xi^{\mu_a}_j \sigma_j \right)^{n-1} \right\rangle \left(1 + \varepsilon \right)  \\
    &\leq 2^m n^m \sum_{\mu_1,\ldots,\mu_m = 1}^M \;  \left \langle\prod_{a=1}^m  (\xi_i^{\mu_a} \sigma_i) \right\rangle \left[\prod_a^m N^{n-1}q^{(n-1) d(\boldsymbol{\xi}^{\mu_a},\boldsymbol{\sigma})} \right](1+\bar{\epsilon}) (1+\varepsilon) \\
    &\leq \langle (\Delta \tilde{E} )^m \rangle (1+o(1)).
    \end{align}

    In line 2, we used \cref{lemma:all k 1} and in line 3, we used \cref{lemma: bulk propagator}.
    
\end{proof}

\subsection{Diagrams and Feynman Rules}

\begin{align*}
    \noindent\resizebox{0.5\textwidth}{!}{%
\begin{tikzpicture}
    \node[treenode, fill=gray]   (R)  at (0,0)      {};
    \node[treenode, fill=gray]   (A)  at (135:2.4)  {};
    \node[treenode, fill=gray]   (B)  at (45:2.4)   {};
    \node[treenode, fill=gray]   (C)  at (315:2.4)  {};
    \node[treenode, fill=gray]   (D)  at (225:2.4)  {};
    \node[treenode, fill=gray]   (A1) at (102:4.6)  {};
    \node[treenode, fill=orange] (A2) at (124:4.6)  {};
    \node[treenode, fill=gray]   (A3) at (146:4.6)  {};
    \node[treenode, fill=gray]   (A4) at (168:4.6)  {};
    \node[treenode, fill=gray]   (B1) at (12:4.6)   {};
    \node[treenode, fill=black]  (B2) at (34:4.6)   {};
    \node[treenode, fill=gray]   (B3) at (56:4.6)   {};
    \node[treenode, fill=gray]   (B4) at (78:4.6)   {};
    \node[treenode, fill=gray]   (C1) at (282:4.6)  {};
    \node[treenode, fill=gray]   (C2) at (304:4.6)  {};
    \node[treenode, fill=gray]   (C3) at (326:4.6)  {};
    \node[treenode, fill=gray]   (C4) at (348:4.6)  {};
    \node[treenode, fill=gray]   (D1) at (192:4.6)  {};
    \node[treenode, fill=gray]   (D2) at (214:4.6)  {};
    \node[treenode, fill=black]  (D3) at (236:4.6)  {};
    \node[treenode, fill=gray]   (D4) at (258:4.6)  {};
    \begin{scope}[on background layer]
        \edgegray{R}{C}
        \foreach \leaf in {A1, A3, A4}     { \edgegray{A}{\leaf} }
        \foreach \leaf in {B1, B3, B4}     { \edgegray{B}{\leaf} }
        \foreach \leaf in {C1, C2, C3, C4} { \edgegray{C}{\leaf} }
        \foreach \leaf in {D1, D2, D4}     { \edgegray{D}{\leaf} }
        \edgepair{A}{A2}
        \edgepair{R}{A}
        \edgepair{R}{B}
        \edgepair{B}{B2}
        \edgepair{R}{D}
        \edgepair{D}{D3}
    \end{scope}
\end{tikzpicture}%
$\quad\quad\quad\quad\quad\quad\quad$
\begin{tikzpicture}
    \node[treenode, fill=gray]   (R)  at (0,0)      {};
    \node[treenode, fill=gray]   (A)  at (135:2.4)  {};
    \node[treenode, fill=gray]   (B)  at (45:2.4)   {};
    \node[treenode, fill=gray]   (C)  at (315:2.4)  {};
    \node[treenode, fill=gray]   (D)  at (225:2.4)  {};
    \node[treenode, fill=gray]   (A1) at (102:4.6)  {};
    \node[treenode, fill=orange] (A2) at (124:4.6)  {};
    \node[treenode, fill=gray]   (A3) at (146:4.6)  {};
    \node[treenode, fill=gray]   (A4) at (168:4.6)  {};
    \node[treenode, fill=gray]   (B1) at (12:4.6)   {};
    \node[treenode, fill=gray]   (B2) at (34:4.6)   {};
    \node[treenode, fill=gray]   (B3) at (56:4.6)   {};
    \node[treenode, fill=gray]   (B4) at (78:4.6)   {};
    \node[treenode, fill=gray]   (C1) at (282:4.6)  {};
    \node[treenode, fill=gray]   (C2) at (304:4.6)  {};
    \node[treenode, fill=gray]   (C3) at (326:4.6)  {};
    \node[treenode, fill=gray]   (C4) at (348:4.6)  {};
    \node[treenode, fill=gray]   (D1) at (192:4.6)  {};
    \node[treenode, fill=black]  (D2) at (214:4.6)  {};
    \node[treenode, fill=black]  (D3) at (236:4.6)  {};
    \node[treenode, fill=gray]   (D4) at (258:4.6)  {};
    \begin{scope}[on background layer]
        \edgegray{R}{B}
        \edgegray{R}{C}
        \foreach \leaf in {A1, A3, A4}     { \edgegray{A}{\leaf} }
        \foreach \leaf in {B1, B2, B3, B4} { \edgegray{B}{\leaf} }
        \foreach \leaf in {C1, C2, C3, C4} { \edgegray{C}{\leaf} }
        \foreach \leaf in {D1, D4}         { \edgegray{D}{\leaf} }
        \edgequad{A}{A2}
        \edgequad{R}{A}
        \edgequad{R}{D}
        \edgepair{D}{D2}
        \edgepair{D}{D3}
    \end{scope}
\end{tikzpicture}%
}
\end{align*}

To each term in ~\cref{eq: expanded gap}, we may associate a \textit{diagram}. Let us use $m=2$ as an illustrative example. When $m=2$, within each term, present are $\boldsymbol{\sigma}$ and two $\boldsymbol{\xi}^{\mu}$, which may or may not be the same leaf memory. Above, two example diagrams are shown. The first diagram on the left corresponds to a case in which $\boldsymbol{\sigma}$ (orange) is a leaf memory and the common ancestors between each pair of $\boldsymbol{\sigma}$, $\boldsymbol{\xi}^{\mu_1}$ (black) and $\boldsymbol{\xi}^{\mu_2}$ (black) are the same (i.e. the root). The diagram on the right corresponds to a case in which the common ancestor between $\boldsymbol{\xi}^{\mu_1}$ and $\boldsymbol{\xi}^{\mu_2}$ is different than the common ancestor between $\boldsymbol{\sigma}$ and either of the $\boldsymbol{\xi}$. We mark the presence of, in the term, a factor of the form $(\xi_i^{\mu_a }\sigma_i)$ with a dashed line which between $\boldsymbol{\sigma}$ and $\boldsymbol{\xi}^{\mu_a}$ following the unique path between them on the tree. Similarly, the factor $\left( \sum_{j\neq i} \xi_j^{\mu_a}\sigma_j\right)^{n-1}$ is represented on the diagram by a solid line which follows the same unique path between $\boldsymbol{\sigma}$ and $\boldsymbol{\xi}^{\mu_a}$. Note that going forward, as the branches of the tree which do not lead to either $\boldsymbol{\sigma}$ or any $\boldsymbol{\xi}$ are unimportant to the diagram weight, we omit them. For example, this will allow us to represent the diagrams from above as tripods.

From the discussion in the previous subsection, we can associate to the lines of a diagram the following weights:
\begin{itemize}
    \item \textbf{Solid line (bulk propagator).} Represents the extensive overlap raised to the $(n{-}1)$th power: $\left(\sum_{j\neq i} \langle\sigma_j \xi^{\mu_1}_j\rangle\right)^{n-1} \to N^{n-1}\, q^{(n-1)\,\ell}$ in the large-$N$ limit. A solid line between nodes $A$ and $B$ carries a factor
    \begin{align}
        \tikz[baseline=-0.5ex]{\draw[thick] (0,0) -- (1,0); \fill (0,0) circle (2pt); \fill (1,0) circle (2pt);} \quad = \quad N^{n-1}\, q^{(n-1)\,d(A,B)}.
    \end{align}

    \item \textbf{Dashed line (single-spin propagator).} Represents the single-site overlap $\langle\sigma_i \xi^{\mu_1}_i\rangle = q^{\ell}$. A dashed line between $A$ and $B$ carries
    \begin{align}
        \tikz[baseline=-0.5ex]{\draw[thick, dashed] (0,0) -- (1,0); \fill (0,0) circle (2pt); \fill (1,0) circle (2pt);} \quad = \quad q^{d(A,B)}.
    \end{align}
\end{itemize}

As mentioned, to calculate $\left\langle \prod_a^m (\xi_i^{\mu_a} \sigma_i) \right\rangle$, we must count the links which appear an odd number of times along the paths from $\boldsymbol{\sigma}$ to each of the $\boldsymbol{\xi}^{\mu_a}$. This manifests in a simple cancellation rule for the dashed lines. Any dashed line passing through multiple nodes of the tree can be factored as the product of the nearest-neighbor dashed lines along the path. This means that two dashed lines, which overlap in any segment, can be reduced to a shorter set of non-overlapping dashed lines. See the following two examples.
\begin{equation*}
\resizebox{0.5\textwidth}{!}{$
\vcenter{\hbox{%
\begin{tikzpicture}[baseline=(current bounding box.center)]
\node[circle, draw=black, fill=black, minimum size=0.5cm] (A) at (0,0) {};
\node[circle, draw=black, fill=black, minimum size=0.5cm] (B) at (1,0) {};
\draw[black, thick, dashed, transform canvas={yshift= 2pt}] (A) -- (B);
\draw[black, thick, dashed, transform canvas={yshift=-2pt}] (A) -- (B);
\end{tikzpicture}%
}}
\;=\;
\vcenter{\hbox{%
\begin{tikzpicture}[baseline=(current bounding box.center)]
\node[circle, draw=black, fill=black, minimum size=0.5cm] (A) at (0,0) {};
\node[circle, draw=black, fill=black, minimum size=0.5cm] (B) at (1,0) {};
\end{tikzpicture}%
}}
\quad\text{       ,        }\quad
\vcenter{\hbox{%
\begin{tikzpicture}[baseline=(current bounding box.center)]
\node[treenode, fill=gray] (R) at (2, 2) {};
\node[treenode, fill=gray] (B1) at (1, 1) {};
\node[treenode, fill=gray] (B2) at (3, 1) {};
\node[treenode, fill=orange] (C1) at (1, 0) {};
\node[treenode, fill=black] (C2) at (2.5, 0) {};
\node[treenode, fill=black] (C3) at (3.5, 0) {};
\begin{scope}[on background layer]
\edgequad{R}{B1}
\edgequad{R}{B2}
\edgequad{B1}{C1}
\edgepair{B2}{C2}
\edgepair{B2}{C3}
\end{scope}
\end{tikzpicture}%
}}
\;=\;
\vcenter{\hbox{%
\begin{tikzpicture}[baseline=(current bounding box.center)]
\node[treenode, fill=gray] (R) at (2, 2) {};
\node[treenode, fill=gray] (B1) at (1, 1) {};
\node[treenode, fill=gray] (B2) at (3, 1) {};
\node[treenode, fill=orange] (C1) at (1, 0) {};
\node[treenode, fill=black] (C2) at (2.5, 0) {};
\node[treenode, fill=black] (C3) at (3.5, 0) {};
\begin{scope}[on background layer]
\edgedouble{R}{B1}
\edgedouble{R}{B2}
\edgedouble{B1}{C1}
\edgepair{B2}{C2}
\edgepair{B2}{C3}
\end{scope}
\end{tikzpicture}%
}}
$}
\end{equation*}
The second example is the same diagram from the right hand side the prior figure, with the middle groups omitted.

For the solid lines, the factor of $N^{n-1}$ arises due to the sum over the spin variables. This sum does not factor as a product over the path, so the factor of $N^{n-1}$ is associated with the \emph{whole line} rather than being broken up along the path. On the other hand the factors of $q^{n-1}$ do factor along the path which is why they come with an exponent of $d(A,B)$. In other words, for any node $C$ on the path between $A$ and $B$ we can factor

\begin{align*}
N^{n-1}\,q^{(n-1)d(A,B)}
    &= \vcenter{\hbox{%
        \begin{tikzpicture}[baseline=(current bounding box.center)]
            \node[circle, draw=black, fill=black, minimum size=0.5cm, label=below:$A$] (A) at (0,0) {};
            \node[circle, draw=black, fill=black, minimum size=0.5cm, label=below:$B$] (B) at (2,0) {};
            \draw[black, thick] (A) -- (B);
        \end{tikzpicture}}} \\
    N^{n-1}\bigl(q^{(n-1)d(A,C)}\cdot q^{(n-1)d(C,B)}\bigr) &= \vcenter{\hbox{%
        \begin{tikzpicture}[baseline=(current bounding box.center)]
            \node[circle, draw=black, fill=black, minimum size=0.5cm, label=below:$A$] (A) at (0,0) {};
            \node[circle, draw=black, fill=black, minimum size=0.5cm, label=below:$C$] (C) at (1.5,0) {};
            \node[circle, draw=black, fill=black, minimum size=0.5cm, label=below:$B$] (B) at (3,0) {};
            \draw[black, thick] (A) -- (C) -- (B);
        \end{tikzpicture}}}.
\end{align*}

We now enumerate the Feynman rules for evaluation of $\langle (\Delta \tilde{E})^m \rangle$.

\begin{enumerate}
    \item Indicate the probe point $\boldsymbol{\sigma}$ on the tree.
    \item Choose $m$ leaf nodes, one for each replica of $\Delta E$, labeled by leaf indices $\mu_1,\dots,\mu_m$.
    \item Draw a dotted and solid line between the probe and each $\boldsymbol{\xi}^{\mu_a}$, going through the tree.
    \item Apply the dotted line cancellation rule repeatedly until no pairs of overlapping dotted lines remain.
    \item Assign propagator weights: $2n N^{n-1}q^{(n-1)\,d(\boldsymbol{\sigma}, \boldsymbol{\xi}^{\mu_a})}$ per solid line and $q^{d(X,Y)}$ per remaining dashed line. Here $X$ and $Y$ range over every pair of nodes necessary to cover all dashed lines, with any covering set yielding the same result.
    \item Sum over all placements of the leaf indices $\mu_1, \dots, \mu_m$. For a homogeneous tree this depends only on the hierarchy of most recent common ancestors of the nodes. Each of those possibilities is associated with a multiplicity that appears as a symmetry factor. It is convenient to group by these factors and multiply by the appropriate combinatorial factor.
\end{enumerate}

\section{First Two Moments} \label{app: first two moments}

\subsection{Diagrams and Combinatorics}

For $\langle \Delta E \rangle$, there is a single leaf index $\mu$. The diagram consists of one solid and one dashed line from $\boldsymbol{\sigma}$ to $\boldsymbol{\xi}^\mu$, contributing $N^{n-1} q^{nd(\boldsymbol{\sigma}, \boldsymbol{\xi}^\mu})$ in total. There is only a single diagram topology in this case and since only one line is present, there are no cancellations. Let the depth of the probe point be $\ell$. The sum over all placements of the leaf indices is equivalent to a sum over the number of diagrams with distance $d(\boldsymbol{\sigma}, \boldsymbol{\xi}^\mu) = d$. Denote this number by $C_\ell(d)$. Then, for our setting with a homogenous tree of height $h$ and degree $K$
\begin{align} \label{eq:counting}
    C_\ell^{(1)}(d) = \begin{cases} 
        (K-1)^{h-\ell},&\quad d = h - \ell \\
        (K-2)(K-1)^{\frac{d-\ell+h}{2} - 1},&\quad h - \ell< d < h + \ell \\
        (K-1)^{h},&\quad d = h + \ell
    \end{cases}
\end{align} 

With this counting, the first moment is
\begin{align} \label{eq: first moment}
    \langle\Delta E\rangle_{\ell} = 2n\, N^{n-1} q^{n(h-\ell)}\sum_{x=0}^{\ell} C^{(1)}_\ell(h-\ell+2x)\, q^{2nx}.
\end{align}

For $m =2$, the diagrams all have a \textit{tripod} topology.
\begin{equation*}
\scalebox{0.5}{
    \begin{tikzpicture}[
  every node/.style = {inner sep=0pt},
  leaf/.style       = {circle, draw=black, line width=0.5pt, minimum size=9mm},
  junction/.style   = {circle, draw=black, fill=white, line width=0.6pt, minimum size=5mm},
  pathnode/.style   = {circle, draw=black!60, fill=gray!55, minimum size=3.5mm},
]

  \node[leaf, fill=orange] at ( 0,  3) (P){};
  \node[leaf, fill=black]  at (-3, -2) (T1){};
  \node[leaf, fill=black]  at ( 3, -2) (T2){};
 
  \node[junction, fill=white] at (0, 0) (J) {};
    \edgedouble{P}{J}
    \edgepair{J}{T1};
    \edgepair{J}{T2};
  \node[pathnode] at (0, 0.75) {};
  \node[pathnode] at (0, 1.5) {};
  \node[pathnode] at (0, 2.25) {};
 
  \node[pathnode] at (-0.6, -0.4)(A){};
  \node[pathnode] at (-1.2, -0.8)(B){};
  \node[pathnode] at (-1.8, -1.2)(C){};
  \node[pathnode] at (-2.4, -1.6)(D){};
 
  \node[pathnode] at (0.75, -0.5) {};
  \node[pathnode] at (1.5,  -1.0) {};
  \node[pathnode] at (2.25, -1.5) {};
    \node[leaf, fill=orange] at ( 0,  3) (P){};
  \node[leaf, fill=black]  at (-3, -2) (T1){};
  \node[leaf, fill=black]  at ( 3, -2) (T2){};
 
  \node[junction, fill=white] at (0, 0) (J) {};

\end{tikzpicture}}
\end{equation*}

Every such diagram can be labeled by a tuple $(s, t_1, t_2) \in \mathbb{Z}_{\geq 0}^3$, where each element of the tuple is the distance between the probe, target 1, or target 2, and the junction which is the median point between the three nodes. Then, given the position of the probe $\boldsymbol{\sigma}$, we must count all such tripods. For a tuple $(s,t_1,t_2)$ we label the count $C_{\ell}((s,t_1,t_2))$. This yields
\begin{align}
    \langle (\Delta \tilde{E})^2 \rangle_{\ell} &=   4n^2 N^{2n-2}\sum_{(s,t_1,t_2)} C^{(2)}_{\ell}((s,t_1,t_2)) q^{n(2s+t_1+t_2) - 2s}
\end{align}

The coefficient of variation will depend on the square of the mean. However, the square of the mean, just like the second moment, will depend on a sum over all positions of two targets. That is, expanded out, the square of the mean will take a form much like that of the second moment, in that it will have the exact same combinatoric structure.
\begin{align}
    \langle \Delta \tilde{E} \rangle_\ell^2 &= 4n^2 N^{2n-2} \sum_{d_1,d_2} C_\ell^{(1)}(d_1) C_\ell^{(1)}(d_2) q^{n(d_1 + d_2)} \\
    &= 4n^2 N^{2n-2}\sum_{(s,t_1,t_2)} C^{(2)}_{\ell}((s,t_1,t_2)) q^{n(2s+t_1+t_2) }
\end{align}
Here, for any junction placement, $d_1 = s + t_1$ and $d_2 = s+t_2$ and so $d_1+d_2= 2s+t_1+t_2$. The only difference between this expression and that of the second moment is the factor of $q^{-2s}$ in the second moment. This allows us to write the variance in a more general form:
\begin{align} \label{eq: variance}
    \langle \Delta \tilde{E}^2 \rangle_\ell - \langle \Delta \tilde{E} \rangle_\ell^2 &= 4n^2 N^{2n-2} \nonumber \\
    &\times \sum_{(s,t_1,t_2)} C^{(2)}_{\ell}((s,t_1,t_2)) q^{n(2s+t_1+t_2) -2s}(1-q^{2s})
\end{align}

Dividing the variance by the mean squared, and taking a square root yields the coefficient of variation.

In the above formulas, we have used the approximate formulas for the both the second moment and the mean squared. The relative errors we determined for the moments do not readily transfer to a relative error on the variance. However, through some additional work, we obtain a similar result. That is,

\begin{lem} [Relative Error on Variance] \label{lemma: variance rel error}

For $q^{2h} = \omega(n/\sqrt{N})$, 
\begin{align}
    \text{Var}(\Delta E) \leq 4 n^2 N^{2(n-1)} (1+o(1))\sum_{\mu_1 \mu_2 }^M  \;
    q^{n(2s+t_1+t_2)-2s}
     (1-q^{2s}).
\end{align}
    
\end{lem}

The proof of this is shown in \cref{sec: proof of variance rel error}. The proof involves many near identical manipulations to that of \cref{lemma:all k 1} and \cref{lemma: bulk propagator}.

The exact form of $C_{\ell}^{(2)}$ will depend on whether $\boldsymbol{\sigma}$ is a leaf, an internal node or the central root. We present the countings below. Note that if $t_1 = t_2 = 0$, the targets overlap and the probe is connected to the target via two solid lines. On the other hand if $s = t_1 = 0$, then the probe and one of the targets overlaps and the probe is connected to the other target via a solid line and a dashed line. For readability, if the subscript from $t$ is omitted, then $t_1 = t_2$. Additionally, the $t_1 \neq t_2$ cases are symmetric. 

If $\boldsymbol{\sigma}$ is itself a leaf, i.e. $\ell = h$, 
{\small
\begin{equation}
C^{(2)}_{\ell=h} =
\begin{dcases}
(K-1)^{2h-1}(K-2),
  & \begin{subarray}{l} h=s=t \end{subarray} \\[2pt]
(K-2)(K-3)(K-1)^{2s-2},
  & \begin{subarray}{l} s\in[1,h-1],\ t=s \end{subarray} \\[2pt]
(K-2)(K-1)^{h+t-1},
  & \begin{subarray}{l} s\in[h+1,2h-1],\\ s+t=2h \end{subarray} \\[2pt]
(K-2)^2(K-1)^{(t_2+3s-4)/2},
  & \begin{subarray}{l} s\in[1,h-1],\ t_1=s,\\
       t_2\in[s+2,\,s+4,\\
       \quad\ \dots,\,s+2(h-s-1)] \end{subarray} \\[2pt]
(K-2)(K-1)^{s+h-1},
  & \begin{subarray}{l} s\in[1,h-1],\\ t_1=s,\ t_2=2h-s \end{subarray} \\[2pt]
(K-2)^2(K-1)^{(s+3t)/2-2},
  & \begin{subarray}{l} t\in[1,h-1],\\ s \in [t+2,t+4,\dots, \\ 2h-t-2] \end{subarray} \\[2pt]
(K-2)(K-1)^{t_2/2-1},
  & \begin{subarray}{l} t_2\in[2,4,\dots,2(h-1)],\\ s=t_1=0 \end{subarray} \\[2pt]
(K-2)(K-1)^{s/2-1},
  & \begin{subarray}{l} s\in[2,4,\dots,2(h-1)],\\ t=0 \end{subarray} \\[2pt]
(K-1)^h,
  & \begin{subarray}{l} 2h=t_2,\ t_1=s=0 \end{subarray} \\[2pt]
(K-1)^h,
  & \begin{subarray}{l} 2h=s,\ t=0 \end{subarray} \\[2pt]
1,
  & \begin{subarray}{l} s=t_1=t_2=0 \end{subarray}
\end{dcases}.
\end{equation}
}

When $0<\ell<h$, the combinatorics are
{\small
\begin{equation} \label{eq: 0 < l < h combinatorics}
C^{(2)}_{0<\ell<h} =
\begin{dcases}
(K-1)^{2t-1}(K-2),
  & \begin{subarray}{l} s=0,\ t=h-\ell \end{subarray} \\[2pt]
(K-1)^{s+2t-1}(K-2),
  & \begin{subarray}{l} s\in[1,h-\ell-1],\\ s+t=h-\ell \end{subarray} \\[2pt]
(K-2)(K-3)(K-1)^{2t-2},
  & \begin{subarray}{l} s\in[1,\ell-1],\\ t-s=h-\ell \end{subarray} \\[2pt]
(K-2)(K-1)^{2h-1},
  & \begin{subarray}{l} s=\ell,\ t=h \end{subarray} \\[2pt]
(K-1)^{h+t-1}(K-2),
  & \begin{subarray}{l} s\in[\ell+1,h+\ell-1],\\ s+t=h+\ell \end{subarray} \\[2pt]
  (K-2)^2(K-1)^{(s+3t+h-\ell)/2-2},
  & \begin{subarray}{l} t\in[1,h-1],\\ s \in [|h-\ell-t|+2,|h-\ell-t|+4,\dots, \\ h+\ell-t-2] \end{subarray} \\[2pt]
(K-2)^2(K-1)^{(t_2+3t_1)/2-2},
  & \begin{subarray}{l} s\in[1,\ell-1],\ t_1-s=h-\ell,\\
       t_2-t_1\in[2,\dots,2(\ell-s-1)] \end{subarray} \\[2pt]
(K-2)(K-1)^{(t_2+3t_1)/2-1},
  & \begin{subarray}{l} s=0,\ t_1=h-\ell,\\
       t_2\in[h-\ell,\,h-\ell+2,\\
       \quad\ \dots,\,h-\ell+2(\ell-1)] \end{subarray} \\[2pt]
(K-2)(K-1)^{h+t_1-1},
  & \begin{subarray}{l} s\in[1,\ell-1],\ t_1-s=h-\ell,\\
       t_2-t_1=2(\ell-s) \end{subarray} \\[2pt]
(K-1)^{t_1+h},
  & \begin{subarray}{l} s=0,\ t_1=h-\ell,\ t_2=h+\ell \end{subarray} \\[2pt]
(K-1)^{h-\ell},
  & \begin{subarray}{l} s=h-\ell,\ t=0 \end{subarray} \\[2pt]
(K-2)(K-1)^{(s+h-\ell)/2-1},
  & \begin{subarray}{l} s\in[h-\ell+2,\,h-\ell+4,\\
       \quad\ \dots,\,h-\ell+2(\ell-1)],\\ t=0 \end{subarray} \\[2pt]
(K-1)^h,
  & \begin{subarray}{l} s=h+\ell,\ t=0 \end{subarray}
\end{dcases}.
\end{equation}
}

Finally, when $\ell = 0$, the combinatorics are
\begin{align}
    C^{(2)}_{\ell = 0} = \begin{cases}
        K(K-1)^{s + 2t-2}(K-2),&\quad s \in[1,h-1], s+t = h\\
        K(K-1)^{2h-1},&\quad s = 0, t = h \\
        K(K-1)^{h-1},&\quad s = h, t = 0
    \end{cases}.
\end{align}

With these exact combinatorics calculated, we will now make the simplifying assumption that $K \gg 3$. This is a reasonable approximation with large training data sets, and will be a valid upper bound on the variance. This will allow us to collapse many of the expressions above to much simpler, easier to manipulate forms, as everywhere we make the substitution $(K-1),(K-2),(K-3) \to K$.

The simplified combinatorial coefficients are, for the first moment,
\begin{align}
    \Tilde{C}_\ell^{(1)}  =  K^{\frac{d+h-\ell}{2}},\quad h-\ell \leq d \leq h+\ell.
\end{align}

For the second moment, when $\ell = h$, we obtain 

{\small
\begin{equation} \label{eq:l = h coefficients}
\tilde{C}^{(2)}_{\ell=h} =
\begin{dcases}
K^{(s+3t)/2},
  & \begin{subarray}{l} t\in[0,h],\\ s\in[t,\,t+2,\,\dots,\,2h-t] \end{subarray} \\[2pt]
K^{(t_2+3t_1)/2},
  & \begin{subarray}{l} t_1=s\in[0,h-1],\\ t_2\in[s+2,\,s+4,\,\dots,\,2h-s] \end{subarray}
\end{dcases}
\end{equation}
}

When $0 < \ell < h$, we get
{\small
\begin{equation} \label{eq: simplified interior combinatorics}
\Tilde{C}_{0<\ell<h}^{(2)} =
\begin{dcases}
K^{2(h-\ell) + 2s},
  & \begin{subarray}{l} s\in[1,\ell],\ t=h-\ell+s \end{subarray} \\[2pt]
K^{2(h-\ell) - s},
  & \begin{subarray}{l} s\in[0,h-\ell],\\ s+t=h-\ell \end{subarray} \\[2pt]
K^{h+t-c},
  & \begin{subarray}{l} t \in [1,h-1], c \in[0, \min(\ell -1, h-t-1)] \\
  s = \ell + h - t - 2c, \end{subarray} \\[2pt]
K^{(t_2+3t_1)/2},
  & \begin{subarray}{l} s\in[0,\ell-1],\ t_1=h-\ell+s,\\
       t_2\in[t_1+2,\,t_1+4,\\
       \quad\ \dots,\,2h-t_1] \end{subarray} \\[2pt]
K^{(s+h-\ell)/2},
  & \begin{subarray}{l} s\in[h-\ell+2,\\
       \quad\ \dots,\,h+\ell] \end{subarray}
\end{dcases}
\end{equation}
}

And in the last case, when $\ell = 0$, we get
\begin{align}
    \Tilde{C}_{\ell = 0}^{(2)} &= \begin{cases}
K^{2h-s}, & s \in [0,h],\ s+t = h 
\end{cases}
\end{align}

\subsection{Coefficient of Variation and Stability}

Recall that the statistics for the energy gap that we require to determine stability for a pattern is given by the mean, and the coefficient of variation. The mean energy gap must be positive, as a pattern can only be stable if perturbing it increases the energy. As explained in the main text, from the form of eq. ~\cref{eq: first moment}, we observe that for prototypes of any $\ell$, this is true. As a result, the relative stabilities of prototypes at different $\ell$ may be understood by studying the coefficient of variation, $\mathcal{C}^v_\ell$. With the combinatorial coefficients determined, we use the approximate formula for the variance given by \cref{lemma: variance rel error} to obtain an upper bound for $(\mathcal{C}^v_\ell)^2$. For the mean squared which is in the denominator of $(\mathcal{C}^v_\ell)^2$, one danger we must be cognizant of is that we cannot use the simplified combinatorics for $\Tilde{C}_\ell^{(1)}$, as this will obscure the relation in ~\cref{eq: variance}.

We may generically write this as
\begin{equation}\label{eq: cv 0 < l < h}
(\mathcal{C}^v_\ell)^2 \;\leq\; \frac{\displaystyle \,\sum_{(s,t_1,t_2)} \Tilde{C}^{(2)}_{\ell}\left( (s,t_1,t_2) \right) q^{n(2s+t_1+t_2) -2s}(1-q^{2s})}
{\displaystyle \left( \sum_{d=h-\ell}^{h+\ell} C^{(1)}_\ell(d)\, q^{n\,d}\right)^2}(1+o(1)),
\end{equation}
where the $o(1)$ error is from \cref{lemma: variance rel error}. This expression is plotted in ~\cref{fig:setting3} for various $\ell$. 

The rigorous upper bound on the coefficient of variation allows us to use Chebyshev's inequality to obtain a bound on the probability that the energy gap is negative for some spin flip. That is,
\begin{align}
    \mathbb{P}\left( \bigcup_{i = 1}^N \Delta E_i < 0 \right) &\leq \sum_i^N \mathbb{P}\left( \Delta E < 0\right) \\
    &\leq N \mathbb{P}\left( |\Delta E - \langle \Delta E \rangle| > \langle \Delta E \rangle \right) \\
    &\leq N (C^v_\ell)^2.
\end{align}

\subsubsection{Prototype Reconstruction $0 < \ell < h$}

Stability of prototypes with $0 < \ell < h$ indicates prototype reconstruction, which we take to be a form of generalization. When $K$ is large, holding other parameters constant, then each sum in ~\cref{eq: cv 0 < l < h} involving a coefficient in ~\cref{eq: simplified interior combinatorics} which contributes a positive power of $K$ is dominated by its largest term. In both the numerator and denominator, this largest term goes as $K^h$, which cancels out and so the limiting value becomes a constant:
\begin{align} \label{eq: cv l<h large K}
    (\mathcal{C}^v_{0 < \ell < h} )^2\;\xrightarrow{\,K\to\infty\,}\; {\,q^{-2\ell}-1\,}.
\end{align}
This expression only vanishes as $q \to 1$. However, in the $q \to 1$ limit, all levels of the tree merge as the correlations become perfect, and the examples only serve to reinforce the recall of the single root memory.

However, as we see in \cref{fig:setting3}, up to a point, $\mathcal{C}^v_{0<\ell<h}$ decreases with $K$. This occurs in a particular regime of $K$, namely $q^{-2} \ll K \ll q^{-2n+1}$. In this regime, for appropriate choice of scaling $K, n$ and $h$ with $N$, we show that the tail probability for there to be a spin, which when flipped lowers the energy, vanishes.

\begin{theo} [Coefficient of Variation for Ancestor Prototypes] \label{theorem: prototype reconstruction}
For $q^{2h} = \omega(2n/\sqrt{N})$ and $K \gg 3$, and $q^{-2} \ll K \ll q^{-2n+1}$, 
\begin{align} \label{eq: cv derived}
    (\mathcal{C}_{0 < \ell < h}^v)^2 \leq  \exp \left(\frac{2(h-\ell)}{K-1} \right) \cdot (1-q^2) \cdot \left[
K^2q^{4n-2}+(Kq^2)^{-1}
\right] \cdot (1+\nu) \cdot (1 + o(1)),
\end{align}

where $\nu = O(K^2q^{4n-2} + Kq^{2n} + (Kq^2)^{-1})$ and the $o(1)$ piece results from the relative error on the variance (\cref{lemma: variance rel error}). Furthermore, for $K \propto N \cdot \kappa(N)$, where $\kappa(N)$ is any arbitrarily slow growing function of $N$, there exist positive constants $c_n$ and $c_h$, such that for $n = c_n \log N$ and $h = c_h \log N$, 
\begin{align}
    \mathbb{P}\left( \bigcup_{i = 1}^N \Delta E_i < 0 \right) \leq o(1).
\end{align}

Thus, any given ancestor prototype has a vanishing probability to be unstable.
    
\end{theo}

The proof of this theorem is given in \cref{sec: prototype theorem proof}. This theorem is the formal statement of the prototype reconstruction result of this work. This statement implies that every prototype in the hierarchy has a finite probability to be a stable local minimum in the energy landscape, leading to a finite density of emergent minima of the Hopfield model.

The term in the brackets of \cref{eq: cv derived} primarily controls the coefficient of variation. If we minimize the term in the brackets with respect to $K$, then we obtain $K^*$, i.e. $K^* \propto q^{-4n/3}$.

\subsubsection{Memorization:  $\ell = h$}
Let us now study the stability of the leaf node. As before, we seek to study limiting values. If we first take the large $K$ limit, then by similar arguments to the case above, we obtain
\begin{align}
    (\mathcal{C}_{\ell = h}^v)^2 \;\xrightarrow{\,K\to\infty\,}\; {q^{-2h} - 1},
\end{align}
which is simply ~\cref{eq: cv l<h large K} with the substitution $\ell = h$. Similarly, in the $h \to \infty$ limit, we once again demand that $Kq^{2n-1} < 1$ to prevent divergence. This can be seen simply from the terms corresponding to the $K^{2s}$ coefficient.

In order to obtain a finite expression when taking the $h \to \infty$ limit, we demand that $Kq^{2n-1} < 1$, as otherwise the infinite series in the above expressions will diverge. Unlike the interior prototype case however, if we take the $K \to  \infty$ limit now, assuming $Kq^{2n-1} < 1$, then $\mathcal{C}_{\ell = h}^v$ vanishes. To see this, we observe that in each sum involving a different coefficient from ~\cref{eq:l = h coefficients}, the dominant term is the lowest term of the sum. After carrying out some algebra, one finds that the leading form of $(\mathcal{C}^v_{\ell = h})^2$ has the form
\begin{align}
    \lim_{h \to \infty}  (\mathcal{C}^v_{\ell = h})^2 \propto {K^2 q^{4n-2} + Kq^{4n-4}}.
\end{align}
The above expression at large $K$, due to the assumption that $Kq^{2n-1} \to 0$, vanishes as this condition requires $q^{2n-1}$ to vanish quicker than $K$ diverges.  

\subsubsection{Root Prototype: $\ell = 0$}
Finally, let us analyze the stability of the root node. Here we obtain

\begin{align}
    (\mathcal{C}^v_{\ell = 0})^2 &= {\sum_{s=0}^{h} (K q^{2})^{-s} - K^{-s}}
\end{align}

As all leaves are derived from the same root, they all serve to reinforce it stability. For $Kq^2 > 1$, the sum in the expression above is convergent at large $h$. The $h\to \infty$ limit of the above expression is
\begin{align}
    \lim_{h \to \infty} (\mathcal{C}^v_{\ell = 0})^2 = {\frac{Kq^2}{Kq^2 - 1} - \frac{K}{K - 1}}.
\end{align}
In the large $K$ limit this expression vanishes. Of the cases we have studied this is the only case in which the coefficient of variation vanishes and so, the root is the only absolutely stable memory at fixed $n,h,q$.

The regime where $Kq^2 < 1$ is one where the correlations are sufficiently weak and there are not enough training examples for the network to take advantage of them in order to reinforce the root. In this regime, the coefficient of variation diverges.

\section{Additional Proofs and Derivations}

A central object that will appear repeatedly in the upcoming calculations are expectation values of the following form:
\begin{align}
    \left\langle \sum_{\substack{j_{1,1},\dots,j_{1,n-k_1} \\ \vdots \\ j_{m,1},\dots,j_{1,n-k_m}}}^N \prod_{a = 1}^m \prod_{r = 1}^{n-k_a} \xi_{j_{r,a}}^{\mu_a} \sigma_{j_{r,a}} \right\rangle.
\end{align}

The sum in the expectation value may be split into a contribution which result from all $j$ being distinct, and a contribution from at least two $j$s being the same. The benefit of this is that, upon using linearity of expectation, the expectation value of terms in the first of these sums will factorize.
\begin{align} \label{eq: collision expansion}
    \left\langle \sum_{\substack{j_{1,1},\dots,j_{1,n-k_1} \\ \vdots \\ j_{m,1},\dots,j_{1,n-k_m}}}^N \prod_{a = 1}^m \prod_{r = 1}^{n-k_a} \xi_{j_{r,a}}^{\mu_a} \sigma_{j_{r,a}} \right\rangle &= \sum_{\text{all }j_{r,a}\text{ distinct}} \prod_{a = 1}^m q^{(n-k_a)d(\boldsymbol{\xi}^{\mu_a},\boldsymbol{\sigma})} +\sum_{\geq 2\ j_{r,a}\text{ same}}  \left\langle \prod_{a = 1}^m \prod_{r = 1}^{n-k_a} \xi_{j_{r,a}}^{\mu_a} \sigma_{j_{r,a}} \right\rangle.
\end{align}

In order to facilitate the upcoming proofs, we will make use of the following helper lemma which places bounds on the second term on the RHS. We use the following shorthand notation $d_a = d(\boldsymbol{\xi}^{\mu_a}\boldsymbol{\sigma})$.

\begin{lem} [Bounding Collision Terms]\label{lemma: collisions}
    The second term of \cref{eq: collision expansion} is non-negative and upper bounded by

    \begin{align}
        \sum_{\geq 2\ j_{r,a}\text{ same}}  \left\langle \prod_{a = 1}^m \prod_{r = 1}^{n-k_a} \xi_{j_{r,a}}^{\mu_a} \sigma_{j_{r,a}}\right\rangle \leq \left( \prod_{a = 1}^m q^{(n-k_a)d_a}\right) {N-1 \choose R} R! \cdot  \sum_s^{R}    \left(\frac{q^{-2h}}{\sqrt{N-mn}} \cdot \frac{eR}{\log 2} \right)^{s}.
    \end{align}
\end{lem}

\begin{proof}
    Non-negativity is easily proven as the correlations are all positive and so the expectation values must all be greater than or equal to $0$.

    To derive the upper bound, we will decompose the summation into colliding and distinct indices. That is, every term in the summation involves some set of indices among which there are collisions (i.e. the $j_{r,a}$ are the same), and then the remaining indices are all distinct. The colliding indices need not all be the same, but for $j_{r,a}$ in the set of colliding indices, there is at least one other $j_{r',a'}$ for which $j_{r,a} = j_{r',a'}$.

    With this logic, and setting $R = \sum_{a=1}^m n - k_a$, the sum may be written as
\begin{align}
    \sum_{\geq 2\ j_{r,a}\text{ same}}  \left\langle \prod_{a = 1}^m \prod_{r = 1}^{n-k_a} \xi_{j_{r,a}}^{\mu_a} \sigma_{j_{r,a}} \right \rangle &= \sum_s^{R} \sum_{\substack{S \subset \{(r,a)\} \\ : |S| = s}} \sum_{\substack{\pi \vdash S \\ |B| \geq 2 \forall B}} \sum_{\substack{ {u_B}_{B \in \pi} \\
    \{v_{(a,r)}\}_{(a,r) \notin S} \\ \text{ distinct} }} \prod_{B\in \pi} \left\langle \prod_{(a,r) \in B} \xi_{u_B}^{\mu_a} \sigma_{u_B} \right\rangle  \prod_{(a,r) \notin S} \left\langle \xi_{v_{(a,r)}}^{\mu_a} \sigma_{v_{(a,r)}} \right\rangle .
\end{align}

We use $s$ to indicate the number of indices amongst whome there are collisions. Then, we choose a subset $S$ of size $s$ amongst all the tuples $(r,a)$. This subset $S$ is further decomposed into subsets, and among those subsets, $j_{(r,a)}$ is equal. Thus, $S = \sqcup_{B \in \pi} u_B$. Now, we shall upper bound this expression by replacing the expectation over the colliding indices with $1$. With this, and using the dashed line propagator, we obtain
\begin{align}
     \sum_{\geq 2\ j_{r,a}\text{ same}}  \left\langle \prod_{a = 1}^m \prod_{r = 1}^{n-k_a} \xi_{j_{r,a}}^{\mu_a} \sigma_{j_{r,a}} \right \rangle &\leq  \sum_s^{R} \sum_{\substack{S \subset \{(r,a)\} \\ : |S| = s}} \sum_{\substack{\pi \vdash S \\ |B| \geq 2 \forall B}} \sum_{\substack{ {u_B}_{B \in \pi} \\
    \{v_{(a,r)}\}_{(a,r) \notin S} \\ \text{ distinct} }}  \prod_{(a,r) \notin S} q^{d_a} \\
    &= \sum_s^{R} \sum_{\substack{S \subset \{(r,a)\} \\ : |S| = s}} \sum_{\substack{\pi \vdash S \\ |B| \geq 2 \forall B}}  {N-1 \choose R- s + |\pi|} (R- s + |\pi|)!\prod_{(a,r) \notin S} q^{d_a} \\
    &= \left( \prod_{a = 1}^m q^{(n-k_a)d_a}\right) \sum_s^{R} \sum_{\substack{S \subset \{(r,a)\} \\ : |S| = s}} \sum_{\substack{\pi \vdash S \\ |B| \geq 2 \forall B}}  {N-1 \choose R- s + |\pi|} (R- s + |\pi|)!\prod_{(a,r) \in S} q^{-d_a} \\
    &\leq  \left( \prod_{a = 1}^m q^{(n-k_a)d_a}\right) \sum_s^{R}  q^{-2hs} {R \choose s} \sum_{b = 1}^{\lfloor s/2 \rfloor} S_2(s,b)  {N-1 \choose R- s + b} (R- s + b)! 
\end{align}

In the third line, we made the replacement $\prod_{(a,r) \notin S} q^{d_a} = \prod_{a = 1}^m q^{(n-k_a)d_a}  \prod_{(a,r) \in S} q^{-d_a}$. In the fourth line, we use the fact that $q^{-d_a} \leq q^{-2h}$ for any $d_a$. We also replace the sums with appropriate counts as the sums no longer depend on any other properties of the subsets beyond their size. $S_r$ denotes the $r$-associated Stirling number of the second kind. 

Next we use the fact that $s - b \geq s/2$ to bound the product of binomial coefficients:
\begin{align}
    {N-1 \choose R- s + b} (R- s + b)! \leq {N-1 \choose R} R! \left( \frac{1}{N-R} \right)^{s/2}.
\end{align}

Substituting, we get
\begin{align}
    &\leq  \left( \prod_{a = 1}^m q^{(n-k_a)d_a}\right) {N-1 \choose R} R! \cdot  \sum_s^{R}  q^{-2hs} {R \choose s} \left( \frac{1}{N-R} \right)^{s/2}  \sum_{b = 1}^{\lfloor s/2 \rfloor} S_2(s,b)   \\
    &\leq  \left( \prod_{a = 1}^m q^{(n-k_a)d_a}\right) {N-1 \choose R} R! \cdot  \sum_s^{R}    \left(\frac{q^{-2h}}{\sqrt{N-R}} \cdot \frac{eR}{\log s} \right)^{s}
\end{align}

To obtain the second line, we used the partition bound on $S_2$ and the binomial bound, and collected terms. Next, we use the fact that $R \leq mn $ to lower bound the $\sqrt{N-R}$ term, and that $\log s \geq \log 2$, which yields
\begin{align}
    &\leq  \left( \prod_{a = 1}^m q^{(n-k_a)d_a}\right) {N-1 \choose R} R! \cdot  \sum_s^{R}    \left(\frac{q^{-2h}}{\sqrt{N-mn}} \cdot \frac{eR}{\log 2} \right)^{s}
\end{align}
    
\end{proof}

\subsection{Proof of \cref{lemma:all k 1}} \label{sec: lemma2 proof}

\begin{proof}

With \cref{eq:T}, we write the $m$th moment as
\begin{align}
    \langle (\Delta E)^m \rangle 
    &= 2^m \sum_{\mu_1,\ldots,\mu_m = 1}^M \; T_{(1,1,\dots 1)}\left( 1 + \sum_{\substack{k_1,\ldots,k_m \\ \rm all\ odd \\
    \neq (1,1\dots,1)}}^n \frac{T_{\mathbf{k}}}{T_{(1,1,\dots 1)}} \right) .
\end{align}

The term in the parenthesis is the relative error. As the correlations are all positive, notice that $T_{\mathbf{k}} \geq 0$ for any $\mathbf{k}$. Thus the relative error is always positive, which yields the lower bound of the lemma. 

We can factorize $T_\mathbf{k}$ to yield
\begin{align}
    T_\mathbf{k} &= \prod_{a=1}^m {n \choose k_a}\, \left \langle\prod_{a=1}^m  (\xi_i^{\mu_a} \sigma_i) \right\rangle \left\langle \prod_{a=1}^m \left(\sum_{j \neq i} \xi^{\mu_a}_j \sigma_j \right)^{n-k_a} \right\rangle \\
    &= \prod_{a=1}^m {n \choose k_a}\, \left \langle\prod_{a=1}^m  (\xi_i^{\mu_a} \sigma_i) \right\rangle \left\langle \sum_{\substack{j_{1,1},\dots,j_{1,n-k_1} \\ \vdots \\ j_{m,1},\dots,j_{1,n-k_m}}}^N \prod_{a = 1}^m \prod_{r = 1}^{n-k_a} \xi_{j_{r,a}}^{\mu_a} \sigma_{j_{r,a}} \right\rangle.
\end{align}

Following \cref{eq: collision expansion}, we obtain

\begin{align}
    &= \prod_{a=1}^m {n \choose k_a}\, \left \langle\prod_{a=1}^m  (\xi_i^{\mu_a} \sigma_i) \right\rangle \left[  \sum_{\text{all }j_{r,a}\text{ distinct}} \prod_{a = 1}^m q^{(n-k_a)d_a} +\sum_{\geq 2\ j_{r,a}\text{ same}}  \left\langle \prod_{a = 1}^m \prod_{r = 1}^{n-k_a} \xi_{j_{r,a}}^{\mu_a} \sigma_{j_{r,a}} \right\rangle \right]
\end{align}

We use the shorthand $d(\boldsymbol{\xi}^{\mu_a},\boldsymbol{\sigma}) = d_a$.

Let us now write the ratio between an arbitrary $\mathbf{k}$ case and the case of $\mathbf{k} = (1,1,\dots,1)$.

\begin{align}
    \varepsilon_\mathbf{k} = \frac{T_{\mathbf{k}}}{T_{(1,1,\dots,1)}} &= \frac{\prod_{a=1}^m {n \choose k_a}\, \left \langle\prod_{a=1}^m  (\xi_i^{\mu_a} \sigma_i) \right\rangle \left[  \sum_{\text{all }j_{r,a}\text{ distinct}} \prod_{a = 1}^m q^{(n-k_a)d_a} +\sum_{\geq 2\ j_{r,a}\text{ same}}  \left\langle \prod_{a = 1}^m \prod_{r = 1}^{n-k_a} \xi_{j_{r,a}}^{\mu_a} \sigma_{j_{r,a}} \right\rangle \right]}{\prod_{a=1}^m {n \choose 1}\, \left \langle\prod_{a=1}^m  (\xi_i^{\mu_a} \sigma_i) \right\rangle \left[  \sum_{\text{all }j_{r,a}\text{ distinct}} \prod_{a = 1}^m q^{(n-1)d_a} +\sum_{\geq 2\ j_{r,a}\text{ same}}  \left\langle \prod_{a = 1}^m \prod_{r = 1}^{n-1} \xi_{j_{r,a}}^{\mu_a} \sigma_{j_{r,a}} \right\rangle \right]} 
\end{align}

To simplify notation, set $R = \sum_{a = 1}^m n-k_a$. We use \cref{lemma: collisions} now in two ways. First, we may remove the second term from the denominator, as it is positive, to obtain a valid upper bound. Second, we use \cref{lemma: collisions} to upper bound the second term in the numerator. We additionally cancel the $\left \langle\prod_{a=1}^m  (\xi_i^{\mu_a} \sigma_i) \right\rangle$ from the numerator and denominator in the prefactor, and use the binomial bound, to obtain

\begin{align}
    \varepsilon_\mathbf{k} &\leq \frac{(en)^{\sum_a k_a}}{n^m \prod_a^m k_a^{k_a}} \cdot \frac{  \sum_{\text{all }j_{r,a}\text{ distinct}} \prod_{a = 1}^m q^{(n-k_a)d_a} +\left( \prod_{a = 1}^m q^{(n-k_a)d_a}\right) {N-1 \choose R} R! \cdot  \sum_s^{R}    \left(\frac{q^{-2h}}{\sqrt{N-mn}} \cdot \frac{eR}{\log 2} \right)^{s} }{  \sum_{\text{all }j_{r,a}\text{ distinct}} \prod_{a = 1}^m q^{(n-1)d_a}  }
\end{align}

We next insert the appropriate counting factors for the sum over distinct indices. 

\begin{align}
    \varepsilon_\mathbf{k} &\leq \frac{(en)^{\sum_a k_a}}{n^m \prod_a^m k_a^{k_a}} \cdot \frac{ {N-1 \choose R}R! \prod_{a = 1}^m q^{(n-k_a)d_a} +\left( \prod_{a = 1}^m q^{(n-k_a)d_a}\right) {N-1 \choose R} R! \cdot  \sum_s^{R}    \left(\frac{q^{-2h}}{\sqrt{N-mn}} \cdot \frac{eR}{\log 2} \right)^{s} }{  {N-1 \choose m(n-1)}(m(n-1))! \prod_{a = 1}^m q^{(n-1)d_a}  } \\
    &= \frac{(en)^{\sum_a k_a}}{n^m \prod_a^m k_a^{k_a}} \cdot \left(\frac{{N-1 \choose R} R!}{{N-1 \choose m(n-1)} (m(n-1))!} \prod_{a=1}^m q^{(1-k_a)d_a} \right) \cdot \left( 1+  \sum_s^{R}    \left(\frac{q^{-2h}}{\sqrt{N-mn}} \cdot \frac{eR}{\log 2} \right)^{s}\right)\\
    &\leq \frac{(en)^{\sum_a k_a}}{n^m \prod_a^m k_a^{k_a}} \cdot \left( \frac{   (1/q)^{2h} }{(N - mn)} \right)^{\sum_a k_a - m} \cdot \left( 1+  \sum_s^{R}    \left(\frac{q^{-2h}}{\sqrt{N-mn}} \cdot \frac{eR}{\log 2} \right)^{s}\right) \label{eq: this line}
\end{align}

In the third line, we used the fact that $q^{-d_a} \leq q^{-2h}$ and the following:
\begin{align} \label{eq:combinatoric manipulation}
    \frac{ {N-1 \choose R}R!}{{N-1 \choose m(n-1)} ( m(n-1) )!} &= \frac{(N-1-(m(n-1))!}{(N-1-R)!} \\
    &= \frac{1}{(N-m(n-1))(N-m(n-1)+1)\dots(N-1-R)} \\
    &\leq \frac{1}{(N-m(n-1))^{\sum_a k_a - m}} \\
    &\leq \frac{1}{(N-mn)^{\sum_a k_a - m}}.
\end{align}

The prefactor in \cref{eq: this line} vanishes under the assertion that $q^{2h} = \omega(n/N)$. 
The geometric series in the brackets of \cref{eq: this line} converges so long as the common ratio is less than $1/2$. This will be true under the stronger assertion that $q^{2h} = \omega(n / \sqrt{N})$. Combined with the prefactor, the entire ratio becomes $o(1)$.

Let us define $\delta = \frac{ n  (1/q)^{2h} }{(N - mn)}$. Then, the overall ratio $\varepsilon_\mathbf{k}$ is now upper bounded by
\begin{align}
    \varepsilon_\mathbf{k} \leq C \delta^{-m} \prod_a^m \left( \frac{e}{k_a} \delta\right)^{k_a} 
\end{align}
where here, $C$ is some constant which depends on the sum in the brackets of \cref{eq: this line}.  
We use this and return to the moment:
\begin{align}
    \langle (\Delta E)^m \rangle &=  2^m \sum_{\mu_1,\ldots,\mu_m = 1}^M \; T_{(1,1,\dots 1)}\left( 1 + \sum_{\substack{k_1,\ldots,k_m \\ \rm all\ odd \\
    \neq (1,1\dots,1)}}^n \varepsilon_\mathbf{k} \right) \\
    &\leq   2^m \sum_{\mu_1,\ldots,\mu_m = 1}^M \; T_{(1,1,\dots 1)}\left( 1 + C \delta^{-m} \sum_{\substack{k_1,\ldots,k_m \\ \rm all\ odd \\
    \neq (1,1\dots,1)}}^n \prod_a^m \left( \frac{e}{k_a} \delta\right)^{k_a} \right) \\
    &\leq   2^m \sum_{\mu_1,\ldots,\mu_m = 1}^M \; T_{(1,1,\dots 1)}\left( 1 +  O(\delta^2) \right)
\end{align}

As long as $q^{2h} = \omega(mn / \sqrt{N})$, the $O(\delta^2)$ term becomes $o(1)$ and so
\begin{align}
    \langle (\Delta E)^m \rangle &= 2^m \sum_{\mu_1,\ldots,\mu_m = 1}^M \; T_{(1,1,\dots 1)}\left( 1 +  o(1) \right),
\end{align}

i.e., the $m$th moment is well approximated by the all $k_a= 1$ term. 
\end{proof}

\subsection{Proof of \cref{lemma: bulk propagator}} \label{sec: proof of bulk propagator}

\begin{proof}
    Let us first expand the term in the expectation.
    \begin{align}
        \left\langle \prod_a^m \left( \sum_{j \neq i} \xi_j^{\mu_a} \sigma_j \right)^{n-1} \right\rangle &= \left\langle\sum_{\substack{j_{1,1},\dots,j_{1,n-1} \\ \vdots \\ j_{m,1},\dots,j_{1,n-1}}}^N \prod_{a = 1}^m \prod_{r = 1}^{n-1} \xi_{j_{r,a}}^{\mu_a} \sigma_{j_{r,a}} \right\rangle 
    \end{align}
    Following \cref{eq: collision expansion}, we split the sum on the RHS into distinct and collision terms.  Counting the total number of distinct $j_{r,a}$ and using \cref{lemma: collisions}, we obtain
    \begin{align}
        &\leq   {N-1 \choose m(n-1)} ( m(n-1) )! \prod_{a = 1}^m q^{(n-1)d_a} +\left( \prod_{a = 1}^m q^{(n-1)d_a}\right) {N-1 \choose m(n-1)} (m(n-1))! \cdot  \sum_s^{m(n-1)}    \left(\frac{q^{-2h}}{\sqrt{N-mn}} \cdot \frac{em(n-1)}{\log s} \right)^{s} \\
        &\leq   \frac{(N-1)!}{(N-1- m(n-1))!} \prod_{a = 1}^m q^{(n-1)d_a} \left(1 + \sum_s^{m(n-1)}    \left(\frac{q^{-2h}}{\sqrt{N-mn}} \cdot \frac{emn}{\log s} \right)^{s} \right) .
    \end{align}

    We obtain that the relative error $\bar{\epsilon}$ is upper bounded by
    \begin{align}
         \bar{\epsilon} \leq \sum_s^{nm} \left( \frac{q^{-2h} e mn}{\log 2 \sqrt{N-mn}} \right)^{s}.
    \end{align}

    This sum and therefore the relative error is $o(1)$ under the constraint $q^{2h} = \omega(mn / \sqrt{N})$. Subsequently, the entire propagator is
    \begin{align}
        &\leq \frac{(N-1)!}{(N- 1-m(n-1))!} \prod_{a = 1}^m q^{(n-1)d_a} \left(1 + o(1)\right) \\
        &\leq \left[ \prod_{a = 1}^m  N^{n-1} q^{(n-1)d_a} \right]\left(1 + o(1)\right).
    \end{align}

    In the final line, we have just upper bounded $\frac{(N-1)!}{(N- 1- m(n-1))!} $ with $N^{m(n-1)}$.

    As the relative error term is positive, the lower bound is obtained by lower bounding the prefactor. This yields
    \begin{align}
\left\langle \prod_a^m \left( \sum_{j \neq i} \xi_j^{\mu_a} \sigma_j \right)^{n-1} \right\rangle &\geq         \frac{(N-1)!}{(N- 1-m(n-1))!} \prod_{a = 1}^m q^{(n-1)d(\boldsymbol{\xi}^{\mu_a},\boldsymbol{\sigma})}.
    \end{align}
    
\end{proof}

\subsection{Proof of \cref{lemma: variance rel error}} \label{sec: proof of variance rel error}

\begin{proof}

The overall variance may be expressed as 
\begin{align}
    \text{Var}(\Delta E) &= 4 \sum_{\mu_1 \mu_2 }^M \sum_{k_1, k_2,\rm\ odd}^n {n \choose k_1}{n \choose k_2} \left[ \;\left \langle (\xi_i^{\mu_1} \sigma_i) (\xi_i^{\mu_2} \sigma_i) \right\rangle \left\langle  \left(\sum_{j \neq i} \xi^{\mu_1}_j \sigma_j \right)^{n-k_1}  \left(\sum_{j \neq i} \xi^{\mu_2}_j \sigma_j \right)^{n-k_2} \right\rangle \right. \nonumber \\
    &\left. -\left\langle (\xi_i^{\mu_1} \sigma_i  \right\rangle \left\langle \left( \sum_{j\neq i} \xi_j^{\mu_1}\sigma_j\right)^{n-k_1}\right\rangle \left\langle (\xi_i^{\mu_2} \sigma_i  \right\rangle \left\langle \left( \sum_{j\neq i} \xi_j^{\mu_2}\sigma_j\right)^{n-k_2}\right\rangle\right] .
\end{align}

As before, we will attempt to approximate this expression with $k_1 = k_2 = 1$. The relative error there is
\begin{align}
    \tilde{\varepsilon}_\mathbf{k} &= \frac{{n\choose k_1}{n\choose k_2}}{n^2}\frac{\left[ \;\left \langle (\xi_i^{\mu_1} \sigma_i) (\xi_i^{\mu_2} \sigma_i) \right\rangle \left\langle  \left(\sum_{j \neq i} \xi^{\mu_1}_j \sigma_j \right)^{n-k_1}  \left(\sum_{j \neq i} \xi^{\mu_2}_j \sigma_j \right)^{n-k_2} \right\rangle -\left \langle (\xi_i^{\mu_1} \sigma_i) \right\rangle \left\langle (\xi_i^{\mu_2} \sigma_i) \right\rangle \left\langle \left( \sum_{j\neq i} \xi_j^{\mu_1}\sigma_j\right)^{n-k_1}\right\rangle \left\langle \left( \sum_{j\neq i} \xi_j^{\mu_2}\sigma_j\right)^{n-k_2}\right\rangle\right]}{\left[ \;\left \langle (\xi_i^{\mu_1} \sigma_i) (\xi_i^{\mu_2} \sigma_i) \right\rangle\left\langle  \left(\sum_{j \neq i} \xi^{\mu_1}_j \sigma_j \right)^{n-1}  \left(\sum_{j \neq i} \xi^{\mu_2}_j \sigma_j \right)^{n-1} \right\rangle -\left \langle (\xi_i^{\mu_1} \sigma_i) \right\rangle \left\langle (\xi_i^{\mu_2} \sigma_i) \right\rangle \left\langle \left( \sum_{j\neq i} \xi_j^{\mu_1}\sigma_j\right)^{n-1}\right\rangle \left\langle \left( \sum_{j\neq i} \xi_j^{\mu_2}\sigma_j\right)^{n-1}\right\rangle\right]} \\
    &\leq \frac{{n\choose k_1}{n\choose k_2}}{n^2}\frac{
    \left \langle (\xi_i^{\mu_1} \sigma_i) (\xi_i^{\mu_2} \sigma_i) \right\rangle 
    \left\langle  \left(\sum_{j \neq i} \xi^{\mu_1}_j \sigma_j \right)^{n-k_1}  \left(\sum_{j \neq i} \xi^{\mu_2}_j \sigma_j \right)^{n-k_2} \right\rangle}
    {\left[ \;\left \langle (\xi_i^{\mu_1} \sigma_i) (\xi_i^{\mu_2} \sigma_i) \right\rangle\left\langle  \left(\sum_{j \neq i} \xi^{\mu_1}_j \sigma_j \right)^{n-1}  \left(\sum_{j \neq i} \xi^{\mu_2}_j \sigma_j \right)^{n-1} \right\rangle -\left \langle (\xi_i^{\mu_1} \sigma_i) \right\rangle \left\langle (\xi_i^{\mu_2} \sigma_i) \right\rangle \left\langle \left( \sum_{j\neq i} \xi_j^{\mu_1}\sigma_j\right)^{n-1}\right\rangle \left\langle \left( \sum_{j\neq i} \xi_j^{\mu_2}\sigma_j\right)^{n-1}\right\rangle\right]} \\
    &= \frac{{n\choose k_1}{n\choose k_2}}{n^2}\frac{ 
    \left\langle  \left(\sum_{j \neq i} \xi^{\mu_1}_j \sigma_j \right)^{n-k_1}  \left(\sum_{j \neq i} \xi^{\mu_2}_j \sigma_j \right)^{n-k_2} \right\rangle}
    {\left\langle  \left(\sum_{j \neq i} \xi^{\mu_1}_j \sigma_j \right)^{n-1}  \left(\sum_{j \neq i} \xi^{\mu_2}_j \sigma_j \right)^{n-1} \right\rangle } \frac{1}{1-\frac{\left \langle (\xi_i^{\mu_1} \sigma_i) \right\rangle \left\langle (\xi_i^{\mu_2} \sigma_i) \right\rangle \left\langle \left( \sum_{j\neq i} \xi_j^{\mu_1}\sigma_j\right)^{n-1}\right\rangle \left\langle \left( \sum_{j\neq i} \xi_j^{\mu_2}\sigma_j\right)^{n-1}\right\rangle}{\left \langle (\xi_i^{\mu_1} \sigma_i) (\xi_i^{\mu_2} \sigma_i) \right\rangle\left\langle  \left(\sum_{j \neq i} \xi^{\mu_1}_j \sigma_j \right)^{n-1}  \left(\sum_{j \neq i} \xi^{\mu_2}_j \sigma_j \right)^{n-1} \right\rangle}} \\
    &\leq \varepsilon_{\mathbf{k}} \left( 1-\frac{\left \langle (\xi_i^{\mu_1} \sigma_i) \right\rangle \left\langle (\xi_i^{\mu_2} \sigma_i) \right\rangle \left\langle \left( \sum_{j\neq i} \xi_j^{\mu_1}\sigma_j\right)^{n-1}\right\rangle \left\langle \left( \sum_{j\neq i} \xi_j^{\mu_2}\sigma_j\right)^{n-1}\right\rangle}{\left \langle (\xi_i^{\mu_1} \sigma_i) (\xi_i^{\mu_2} \sigma_i) \right\rangle\left\langle  \left(\sum_{j \neq i} \xi^{\mu_1}_j \sigma_j \right)^{n-1}  \left(\sum_{j \neq i} \xi^{\mu_2}_j \sigma_j \right)^{n-1} \right\rangle}\right)^{-1}.
\end{align}

In the second line, we used the fact that the correlations are positive to remove the subtraction from the numerator. Then, by factoring the expression in line 3, we obtain a prefactor which is identical in form to that in \cref{lemma:all k 1}. We upper bound the prefactor using \cref{lemma:all k 1} to get to line 4. Continuing onwards, we now focus on the term in the parenthesis. To upper bound the variance, we must lower bound the term in the parenthesis.

\begin{align}
    1-\frac{\left \langle (\xi_i^{\mu_1} \sigma_i) \right\rangle \left\langle (\xi_i^{\mu_2} \sigma_i) \right\rangle \left\langle \left( \sum_{j\neq i} \xi_j^{\mu_1}\sigma_j\right)^{n-1}\right\rangle \left\langle \left( \sum_{j\neq i} \xi_j^{\mu_2}\sigma_j\right)^{n-1}\right\rangle}{\left \langle (\xi_i^{\mu_1} \sigma_i) (\xi_i^{\mu_2} \sigma_i) \right\rangle\left\langle  \left(\sum_{j \neq i} \xi^{\mu_1}_j \sigma_j \right)^{n-1}  \left(\sum_{j \neq i} \xi^{\mu_2}_j \sigma_j \right)^{n-1} \right\rangle} &= 1-\frac{q^{d(\boldsymbol{\xi}^{\mu_1},\boldsymbol{\sigma}) + d(\boldsymbol{\xi}^{\mu_2},\boldsymbol{\sigma})}\left\langle \left( \sum_{j\neq i} \xi_j^{\mu_1}\sigma_j\right)^{n-1}\right\rangle \left\langle \left( \sum_{j\neq i} \xi_j^{\mu_2}\sigma_j\right)^{n-1}\right\rangle}{q^{d(\boldsymbol{\xi}^{\mu_1}, \boldsymbol{\xi}^{\mu_2})}\left\langle  \left(\sum_{j \neq i} \xi^{\mu_1}_j \sigma_j \right)^{n-1}  \left(\sum_{j \neq i} \xi^{\mu_2}_j \sigma_j \right)^{n-1} \right\rangle} \\
    &= 1-q^{2s}\frac{\left\langle \left( \sum_{j\neq i} \xi_j^{\mu_1}\sigma_j\right)^{n-1}\right\rangle \left\langle \left( \sum_{j\neq i} \xi_j^{\mu_2}\sigma_j\right)^{n-1}\right\rangle}{\left\langle  \left(\sum_{j \neq i} \xi^{\mu_1}_j \sigma_j \right)^{n-1}  \left(\sum_{j \neq i} \xi^{\mu_2}_j \sigma_j \right)^{n-1} \right\rangle} \\
    &\geq 1 - q^{2s}.
\end{align}
In the first line we used the dashed line rule and in the second line, we used the tripod relation, where $s$ is the distance from the probe $\boldsymbol{\sigma}$ to the median point. The third line is obtained by realizing that as all correlations in the problem are positive. Therefore, for even $n$ and due to the independence of the sites, $(\sum_{j\neq i} \xi_j^{\mu_1}\sigma_j)^{n-1})$ and $(\sum_{j\neq i} \xi_j^{\mu_2}\sigma_j)^{n-1})$ are positively associated and the denominator must be at least as large as the numerator (see FKG inequality). Therefore, the fraction is upper bounded by 1. Thus the full term has the upper bound:
\begin{align}
    \tilde{\varepsilon}_{\mathbf{k}} &\leq \frac{\varepsilon_{\mathbf{k}}}{1-q^{2s}}.
\end{align}

Here, $s$ will depend on $\mu_1$ and $\mu_2$. It is possible for $s = 0$. But as we will see, the contribution to the variance from these terms is 0. 

We return now to the full variance formula:
\begin{align}
    \text{Var}(\Delta E) &\leq 4 n^2\sum_{\mu_1 \mu_2 , s\neq 0}^M \left[ \;\left \langle (\xi_i^{\mu_1} \sigma_i) (\xi_i^{\mu_2} \sigma_i) \right\rangle \left\langle  \left(\sum_{j \neq i} \xi^{\mu_1}_j \sigma_j \right)^{n-1}  \left(\sum_{j \neq i} \xi^{\mu_2}_j \sigma_j \right)^{n-1} \right\rangle \right. \nonumber \\
    &\left. -\left\langle (\xi_i^{\mu_1} \sigma_i  \right\rangle \left\langle \left( \sum_{j\neq i} \xi_j^{\mu_1}\sigma_j\right)^{n-1}\right\rangle \left\langle (\xi_i^{\mu_2} \sigma_i  \right\rangle \left\langle \left( \sum_{j\neq i} \xi_j^{\mu_2}\sigma_j\right)^{n-1}\right\rangle\right] \times \left(1 + \frac{1}{1-q^{2s(\mu_1,\mu_2)}}\sum_{k_1,k_2 \neq (1,1),\rm odd}^n \varepsilon_\mathbf{k}\right) \\
    &\leq 4 n^2\sum_{\mu_1 \mu_2 , s\neq 0}^M \left[ \;\left \langle (\xi_i^{\mu_1} \sigma_i) (\xi_i^{\mu_2} \sigma_i) \right\rangle \left\langle  \left(\sum_{j \neq i} \xi^{\mu_1}_j \sigma_j \right)^{n-1}  \left(\sum_{j \neq i} \xi^{\mu_2}_j \sigma_j \right)^{n-1} \right\rangle \right. \nonumber \\
    &\left. -\left\langle (\xi_i^{\mu_1} \sigma_i  \right\rangle \left\langle \left( \sum_{j\neq i} \xi_j^{\mu_1}\sigma_j\right)^{n-1}\right\rangle \left\langle (\xi_i^{\mu_2} \sigma_i  \right\rangle \left\langle \left( \sum_{j\neq i} \xi_j^{\mu_2}\sigma_j\right)^{n-1}\right\rangle\right] \times \left(1 + \frac{1}{1-q^{2s(\mu_1,\mu_2)}}O(\delta^2)\right) \\
    &= 4 n^2 \sum_{\mu_1 \mu_2, s\neq 0 }^M \left[ \;q^{t_1+t_2} \left\langle  \left(\sum_{j \neq i} \xi^{\mu_1}_j \sigma_j \right)^{n-1}  \left(\sum_{j \neq i} \xi^{\mu_2}_j \sigma_j \right)^{n-1} \right\rangle  -q^{2s+t_1+t_2} \left\langle \left( \sum_{j\neq i} \xi_j^{\mu_1}\sigma_j\right)^{n-1}\right\rangle \left\langle \left( \sum_{j\neq i} \xi_j^{\mu_2}\sigma_j\right)^{n-1}\right\rangle\right] \nonumber \\
    &\times \left(1 + \frac{1}{1-q^{2s(\mu_1,\mu_2)}}O(\delta^2)\right)
\end{align}

We must now control the term in the bracket. We will show that this term, divided by using the solid line approximation is 1 plus a small error term. 

\begin{align}
    &\frac{q^{t_1+t_2} \left\langle  \left(\sum_{j \neq i} \xi^{\mu_1}_j \sigma_j \right)^{n-1}  \left(\sum_{j \neq i} \xi^{\mu_2}_j \sigma_j \right)^{n-1} \right\rangle  -q^{2s+t_1+t_2} \left\langle \left( \sum_{j\neq i} \xi_j^{\mu_1}\sigma_j\right)^{n-1}\right\rangle \left\langle \left( \sum_{j\neq i} \xi_j^{\mu_2}\sigma_j\right)^{n-1}\right\rangle}
    {q^{t_1+t_2} \left\langle  \left(\sum_{j \neq i} \xi^{\mu_1}_j \sigma_j \right) \right\rangle^{n-1} \left\langle  \left(\sum_{j \neq i} \xi^{\mu_2}_j \sigma_j \right) \right\rangle^{n-1}  -q^{2s+t_1+t_2} \left\langle  \left(\sum_{j \neq i} \xi^{\mu_1}_j \sigma_j \right) \right\rangle^{n-1} \left\langle  \left(\sum_{j \neq i} \xi^{\mu_2}_j \sigma_j \right) \right\rangle^{n-1}} \\
    &= \frac{\left\langle  \left(\sum_{j \neq i} \xi^{\mu_1}_j \sigma_j \right)^{n-1}  \left(\sum_{j \neq i} \xi^{\mu_2}_j \sigma_j \right)^{n-1} \right\rangle  -q^{2s} \left\langle \left( \sum_{j\neq i} \xi_j^{\mu_1}\sigma_j\right)^{n-1}\right\rangle \left\langle \left( \sum_{j\neq i} \xi_j^{\mu_2}\sigma_j\right)^{n-1}\right\rangle}
    {\left\langle  \left(\sum_{j \neq i} \xi^{\mu_1}_j \sigma_j \right) \right\rangle^{n-1} \left\langle  \left(\sum_{j \neq i} \xi^{\mu_2}_j \sigma_j \right) \right\rangle^{n-1}  \left( 1- q^{2s} \right)}  \\
    &= \frac{\left\langle  \left(\sum_{j \neq i} \xi^{\mu_1}_j \sigma_j \right)^{n-1}  \left(\sum_{j \neq i} \xi^{\mu_2}_j \sigma_j \right)^{n-1} \right\rangle}{\left\langle  \left(\sum_{j \neq i} \xi^{\mu_1}_j \sigma_j \right) \right\rangle^{n-1} \left\langle  \left(\sum_{j \neq i} \xi^{\mu_2}_j \sigma_j \right) \right\rangle^{n-1}} \frac{1 - q^{2s} \frac{ \left\langle \left( \sum_{j\neq i} \xi_j^{\mu_1}\sigma_j\right)^{n-1}\right\rangle \left\langle \left( \sum_{j\neq i} \xi_j^{\mu_2}\sigma_j\right)^{n-1}\right\rangle}{\left\langle  \left(\sum_{j \neq i} \xi^{\mu_1}_j \sigma_j \right)^{n-1}  \left(\sum_{j \neq i} \xi^{\mu_2}_j \sigma_j \right)^{n-1} \right\rangle}}{1-q^{2s}} \\
    &\leq (1 + \bar{\epsilon}) \frac{1 - q^{2s} \frac{ \left\langle \left( \sum_{j\neq i} \xi_j^{\mu_1}\sigma_j\right)^{n-1}\right\rangle \left\langle \left( \sum_{j\neq i} \xi_j^{\mu_2}\sigma_j\right)^{n-1}\right\rangle}{\left\langle  \left(\sum_{j \neq i} \xi^{\mu_1}_j \sigma_j \right)^{n-1}  \left(\sum_{j \neq i} \xi^{\mu_2}_j \sigma_j \right)^{n-1} \right\rangle}}{1-q^{2s}} \\
    &\leq (1 + \bar{\epsilon}) \frac{1 - q^{2s} \frac{ \left\langle \left( \sum_{j\neq i} \xi_j^{\mu_1}\sigma_j\right)\right\rangle^{n-1} \left\langle \left( \sum_{j\neq i} \xi_j^{\mu_2}\sigma_j\right)\right\rangle^{n-1}}{\left\langle  \left(\sum_{j \neq i} \xi^{\mu_1}_j \sigma_j \right)\right\rangle^{n-1} \left\langle \left(\sum_{j \neq i} \xi^{\mu_2}_j \sigma_j \right) \right\rangle^{n-1}(1+\bar{\epsilon})}}{1-q^{2s}} \\
    &= (1 + \bar{\epsilon}) \frac{1 - q^{2s} (1+\bar{\epsilon})^{-1}}{1-q^{2s}} \\
    &= 1 + \frac{\bar{\epsilon}}{1-q^{2s}}
\end{align}

In the second line, we canceled out $q^{t_1+t_2}$, in the third line we factorized the expression, and in the fourth and fifth lines, we applied \cref{lemma: bulk propagator}. Finally, in the last line we omitted the $1 - \frac{1}{1+\bar{\epsilon}}$ term.

Once again returning to the full variance, we obtain
\begin{align}
    \text{Var}(\Delta E) &\leq 4 n^2 \sum_{\mu_1 \mu_2, s\neq 0 }^M  \;q^{t_1+t_2} \left\langle \left( \sum_{j\neq i} \xi_j^{\mu_1}\sigma_j\right)\right\rangle^{n-1} \left\langle \left( \sum_{j\neq i} \xi_j^{\mu_2}\sigma_j\right)\right\rangle^{n-1} (1-q^{2s(\mu_1,\mu_2)}) \left(1 + \frac{\bar{\epsilon}}{1-q^{2s(\mu_1,\mu_2)}}\right)  \times  \left(1 + \frac{1}{1-q^{2s(\mu_1,\mu_2)}}O(\delta^2)\right).
\end{align}
The factor $(1-q^{2s})$ is 0 for $s = 0$, so even though the error terms are uncontrolled there, the overall expression is 0. So we are free to include those terms back into the summation. Since those terms are 0, we can upper bound the error terms involving $s$ with $s = 1$. This ultimately yields
\begin{align}
    \text{Var}(\Delta E) &\leq 4 n^2 \left[  \left(1 + \frac{\bar{\epsilon}}{1-q^{2}}\right)  \times  \left(1 + \frac{1}{1-q^{2}}O(\delta^2)\right) \right]\sum_{\mu_1 \mu_2 }^M  \;q^{t_1+t_2} \left\langle \left( \sum_{j\neq i} \xi_j^{\mu_1}\sigma_j\right)\right\rangle^{n-1} \left\langle \left( \sum_{j\neq i} \xi_j^{\mu_2}\sigma_j\right)\right\rangle^{n-1} (1-q^{2s})\\
    &\leq 4 n^2 N^{2(n-1)} (1+o(1))\sum_{\mu_1 \mu_2 }^M  \;
    q^{n(2s+t_1+t_2)-2s}
     (1-q^{2s})
\end{align}

\end{proof}

\subsection{Proof of \cref{theorem: prototype reconstruction}} \label{sec: prototype theorem proof}

\begin{proof}
The expression in the numerator of $\mathcal{C}_{\ell}^v$ is the variance, which in simplified form is
\begin{align} \label{eq: cv numerator}
    \tilde{\text{Var}}[\Delta(E) ]  &= 4n^2N^{2n-2} \sum_{(s,t_1,t_2)} \Tilde{C}^{(2)}_{\ell}\left( (s,t_1,t_2) \right) q^{n(2s+t_1+t_2) -2s}(1-q^{2s}).
\end{align}
We use $ \tilde{\text{Var}}$ to denote the variance approximated using \cref{lemma: variance rel error} and the simplified combinatorics.

 We assume that $q^{-2} \ll K \ll q^{-2n+1}$. Upon some algebraic manipulations, we expand \cref{eq: cv numerator} as
\begin{align}
    \tilde{\text{Var}}[\Delta(E)] &= 4n^2N^{2n-2} (Kq^n)^{2(h-\ell)} \left[ \sum_{s = 1}^\ell (Kq^{2n-1})^{2s}(1-q^{2s}) + \sum_{s = 0}^{h-\ell} (Kq^2)^{-s}(1-q^{2s}) + \sum_{a=1}^{\ell}
        (Kq^{2n-1})^{2a}
        \sum_{u=1}^{h-\ell+a-1}
        (Kq^2)^{-u}
        \left(1-q^{2(a+u)}\right) \right. \nonumber \\
    &\left. + 2
    \sum_{s=1}^{\ell-1}
    (Kq^{2n-1})^{2s}
    (1-q^{2s})
    \sum_{y=1}^{\ell-s}
    (Kq^{2n})^y +
(Kq^2)^{-(h-\ell)}
\sum_{x=1}^{\ell}
(Kq^{4n-4})^x
\left(1-q^{2(h-\ell+2x)}\right) \right]
\end{align}

Each of the sums in the bracket here are geometric series. Furthermore, by the assumption on $K$, the common ratio of each of the sums is less than one. Thus we may upper bound the entire expression by setting the top of each sum to be $\infty$.

Our goal is to obtain an expression multiplied by a term of the form $1+o(1)$. To that end, we will factor out the first term from each sum. Note that each of the sums carries a term of the form $1-q^{2s}$ so this may be factored out of the entire expression. This results in the following:
\begin{align} \label{eq: simplified variance}
    \tilde{\text{Var}}[\Delta(E)] &\leq 4n^2N^{2n-2} (Kq^n)^{2(h-\ell)} (1-q^2) \Bigg[
(K^2q^{4n-2})(1+\lambda_1)
+(Kq^2)^{-1}(1+\lambda_2)
+(Kq^{4n-4})(1+q^2)(1+\lambda_3)
\nonumber\\
&+2(K^3q^{6n-2})(1+\lambda_4)
+(Kq^2)^{-(h-\ell)}(Kq^{4n-4})
\frac{1-q^{2(h-\ell+2)}}{1-q^2}
(1+\lambda_5)
\Bigg], 
\end{align}

where in the above expression, each of the $\lambda_i$ terms are
\begin{align}
    \lambda_1 &:= \sum_{s=2}^{\infty}
        \frac{\big(Kq^{2n-1}\big)^{2s}\big(1-q^{2s}\big)}
             {\big(Kq^{2n-1}\big)^{2}\big(1-q^{2}\big)}
        \;\leq\; \sum_{s=2}^{\infty} s\,\big(K^{2}q^{4n-2}\big)^{s-1}
        \;=\; O\!\big(K^{2}q^{4n-2}\big) \\[6pt]
    \lambda_2 &:= \sum_{s=2}^{\infty}
        \frac{\big(Kq^{2}\big)^{-s}\big(1-q^{2s}\big)}
             {\big(Kq^{2}\big)^{-1}\big(1-q^{2}\big)}
        \;\leq\; \sum_{s=2}^{\infty} s\,\big(Kq^{2}\big)^{-(s-1)}
        \;=\; O\!\big((Kq^{2})^{-1}\big) \\[6pt]
    \lambda_3 &:= \sum_{\substack{a,u\,\geq\,1\\ (a,u)\,\neq\,(1,1)}}
        \frac{\big(Kq^{2n-1}\big)^{2a}\big(Kq^{2}\big)^{-u}\big(1-q^{2(a+u)}\big)}
             {\big(Kq^{2n-1}\big)^{2}\big(Kq^{2}\big)^{-1}\big(1-q^{4}\big)}
        \nonumber\\
    &\;\leq\; \sum_{\substack{a,u\,\geq\,1\\ (a,u)\,\neq\,(1,1)}}
        (a+u)\,\big(K^{2}q^{4n-2}\big)^{a-1}\big(Kq^{2}\big)^{-(u-1)}
        \;=\; O\!\big(K^{2}q^{4n-2}+(Kq^{2})^{-1}\big) \\[6pt]
    \lambda_4 &:= \sum_{\substack{s,y\,\geq\,1\\ (s,y)\,\neq\,(1,1)}}
        \frac{\big(Kq^{2n-1}\big)^{2s}\big(1-q^{2s}\big)\big(Kq^{2n}\big)^{y}}
             {\big(Kq^{2n-1}\big)^{2}\big(1-q^{2}\big)\big(Kq^{2n}\big)}
        \nonumber\\
    &\;\leq\; \sum_{\substack{s,y\,\geq\,1\\ (s,y)\,\neq\,(1,1)}}
        s\,\big(K^{2}q^{4n-2}\big)^{s-1}\big(Kq^{2n}\big)^{y-1}
        \;=\; O\!\big(K^{2}q^{4n-2}+Kq^{2n}\big) \\[6pt]
    \lambda_5 &:= \sum_{x=2}^{\infty}
        \big(Kq^{4n-4}\big)^{x-1}\,
        \frac{1-q^{2(h-\ell+2x)}}{1-q^{2(h-\ell+2)}}
        \;\leq\; \frac{1}{1-q^{6}}\sum_{x=2}^{\infty}\big(Kq^{4n-4}\big)^{x-1}
        \;=\; O\!\big(Kq^{4n-4}\big)
\end{align}

Here, for $\lambda_1$ to $\lambda_4$, we have used the inequality $\frac{1-q^{2s}}{1-q^{2}} \;=\; 1+q^{2}+\cdots+q^{2s-2} \;\leq\; s$. We next notice that in the expression in the brackets of \cref{eq: simplified variance}, each term involves a factor of $K^2q^{4n-2}$ and/or $Kq^2$. We subsequently factor out $[K^2q^{4n-2} + (Kq^2)^{-1}]$. Upon doing so, and also upper bounding $\frac{1-q^{2(h-\ell+2)}}{1-q^2}$ using the earlier inequality, we obtain
\begin{align}
    \tilde{\text{Var}}[\Delta(E)] &\leq
4n^2N^{2n-2} (Kq^n)^{2(h-\ell)}(1-q^2)
\left[
K^2q^{4n-2}+(Kq^2)^{-1}
\right]
\cdot 
\Bigg[
1+
\frac{
K^2q^{4n-2}\lambda_1+(Kq^2)^{-1}\lambda_2
}{
K^2q^{4n-2}+(Kq^2)^{-1}
}
+
\frac{
(1+q^2)(K^2q^{4n-2})(Kq^2)^{-1}(1+\lambda_3)
}{
K^2q^{4n-2}+(Kq^2)^{-1}
}
\nonumber\\
&+
\frac{
2(K^2q^{4n-2})(Kq^{2n})(1+\lambda_4)
}{
K^2q^{4n-2}+(Kq^2)^{-1}
}
+
\frac{
(Kq^2)^{-(h-\ell)}
(K^2q^{4n-2})(Kq^2)^{-1}
(h-\ell+2)(1+\lambda_5)
}{
K^2q^{4n-2}+(Kq^2)^{-1}
}
\Bigg]
\\
&\leq 4n^2N^{2n-2} 
(Kq^n)^{2(h-\ell)}(1-q^2)
\left[
K^2q^{4n-2}+(Kq^2)^{-1}
\right]
\cdot 
\Bigg[
1+\lambda_1+\lambda_2
+(1+q^2)(K^2q^{4n-2})(1+\lambda_3)
+2(Kq^{2n})(1+\lambda_4)
\nonumber\\
&\qquad\qquad
+(h-\ell+2)(Kq^2)^{-(h-\ell)}
(K^2q^{4n-2})(1+\lambda_5)
\Bigg]
\\
&\leq 4n^2N^{2n-2} 
(Kq^n)^{2(h-\ell)}(1-q^2)
\left[
K^2q^{4n-2}+(Kq^2)^{-1}
\right]
(1+\nu).
\end{align}

Here, we set $\nu$ to be
\begin{align}
    \nu &:= \lambda_1+\lambda_2
+(1+q^2)(K^2q^{4n-2})(1+\lambda_3)
+2(Kq^{2n})(1+\lambda_4)
+(h-\ell+2)(Kq^2)^{-(h-\ell)}
(K^2q^{4n-2})(1+\lambda_5) \\
&= O(K^2q^{4n-2} + Kq^{2n} + (Kq^2)^{-1})
\end{align}

With this result in hand, we now determine the denominator to the squared coefficeint of variation, which is simply the mean squared. This quantity is lower bounded by a modification to the approximate formula (\cref{theorem: approximate gap}) using the lower bound in \cref{lemma: bulk propagator}. 

\begin{align}
\left\langle \Delta E \right\rangle_\ell^2
&\geq 
4n^2 \left( \frac{(N-1)!}{(N-n)!}\right)^2
\left((K-1)q^n\right)^{2(h-\ell)}
\times 
\Bigg[
1+
\frac{K-2}{K-1}
\sum_{x=1}^{\ell-1}
\left((K-1)q^{2n}\right)^x
+
\left((K-1)q^{2n}\right)^\ell
\Bigg]^2
\\
&\geq 
4n^2 \left( \frac{(N-1)!}{(N-n)!}\right)^2
\left((K-1)q^n\right)^{2(h-\ell)}
\\
&\geq 
4n^2N^{2n-2} \left(1 - \frac{n(n-1)}{N} \right)
\left((K-1)q^n\right)^{2(h-\ell)}.
\end{align}

When we invert this in the squared coefficient of variation, the term in the parenthesis becomes a $1+o(1)$ factor, which combines with the relative error on the variance to remain $1+o(1)$. With this, we obtain the upper bound on squared coefficient of variation.

\begin{align}
    (\mathcal{C}_{0<\ell<h}^v)^2 &\leq \left( \frac{K}{K-1} \right)^{2(h-\ell)}(1-q^2) \cdot \left[
K^2q^{4n-2}+(Kq^2)^{-1}
\right] \cdot 
(1+\nu) \cdot (1+o(1)) \\
&\leq \exp \left(\frac{2(h-\ell)}{K-1} \right) \cdot \left[
K^2q^{4n-2}+(Kq^2)^{-1}
\right] \cdot (1+\nu) \cdot (1 + o(1)) \label{eq: cv line}
\end{align}

Thus we obtain the inequality of the theorem. 

Now, we set $K = c_K \cdot N\cdot  \kappa(N),\ n = c_n \log N$ and $h = c_h \log N$. Here, $\kappa(N)$ is an arbitrarily slowly growing function of $N$. For this choice of parameter scaling $\nu$ becomes $o(1)$.

First, our formulas are only valid for $q^{2h} = \omega(n/\sqrt{N})$. For this to be satisfied, we require $\lim_{N \to \infty} \frac{c_n \log N}{\sqrt{N} q^{2c_h \log N}} \to 0$, which subsequently requires
\begin{align} \label{eq: this relation}
    c_h \log(1/q) < 1/4.
\end{align}

Return now to \cref{eq: cv line}, we focus on the term in the brackets, as this is what will be vanishing. 

\begin{align} \label{eq: cv final}
    K^2q^{4n-2}+(Kq^2)^{-1} &= c_K^2 N^{2}  q^{4c_n \log N - 2} (\kappa(N))^2+ \frac{1}{c_K N^{1}\kappa(N) q^2}
\end{align}

From \cref{eq: cv final}, we may now use the union bound over all $N$ spins. Note that the exponential prefactor also becomes $1+o(1)$ under this parameter scaling. 

\begin{align}
     \mathbb{P}\left( \bigcup_{i = 1}^N \Delta E_i < 0 \right) &\leq N \cdot \left( c_K^2 N^{2}  q^{4c_n \log N - 2} (\kappa(N))^2+ \frac{1}{c_K N^{1}\kappa(N) q^2} \right)  (1+o(1))
\end{align}

When $N$ multiples the second term in the parenthesis, that term decays with $N$ due to the $\kappa(N)$ term in the denominator. As for the first term in the parenthesis, it too decays so long as 
\begin{align}
    c_n \log(1/q) > \frac{3}{4} 
\end{align}

This relation and \cref{eq: this relation} imply also that $c_n > 3c_h$. This proves the second statement in the theorem. 
\end{proof}

\end{document}